\documentclass[aps,prd,superscriptaddress,preprintnumbers,nofootinbib,10pt,longbibliography]{revtex4-2}

\usepackage[utf8]{inputenc}
\usepackage[textwidth=425pt]{geometry}

\usepackage{amsthm,amsmath,amsfonts,amssymb,dsfont}

\usepackage{braket}
\usepackage{caption,subcaption}

\usepackage{hyperref}

\usepackage[dvipsnames]{xcolor}
\usepackage{tikz, tikz-cd}
\usepackage{float}

\usepackage{mathtools}
\usepackage{CJK}

\newtheoremstyle{thm}
{8pt}
{8pt}
{\itshape}
{}
{\bfseries}
{.}
{.6em}
{}

\newtheoremstyle{conj}
{8pt}
{8pt}
{\itshape}
{}
{\bfseries}
{.}
{.6em}
{}

\newtheoremstyle{def}
{12pt}
{12pt}
{}
{}
{\bfseries}
{.}
{.6em}
{}

\newtheoremstyle{rem}
{8pt}
{8pt}
{}
{}
{\itshape}
{.}
{.6em}
{}

\theoremstyle{thm}
\newtheorem{thm}{Theorem}[section]
\newtheorem{lemma}[thm]{Lemma}
\newtheorem{proposition}[thm]{Proposition}

\newcounter{conj}
\theoremstyle{conj}
\newtheorem{conjecture}[conj]{Conjecture}

\theoremstyle{conj}

\theoremstyle{def}

\newtheorem{definition}[thm]{Definition}
\newtheorem{notation}[thm]{Notation}

\theoremstyle{rem}
\newtheorem{remark}[thm]{Remark}
\newtheorem{outlook}[thm]{Outlook}

\begin{document}

\title{Further Evidence for Near-Tsirelson Bell-CHSH Violations in Quantum Field Theory via Haar Wavelets}

\author{David Dudal}
\affiliation{Department of Physics, KU Leuven Campus Kortrijk -- Kulak,  Etienne Sabbelaan 53, 8500 Kortrijk, Belgium}
\email{david.dudal@kuleuven.be}

\author{Ken Vandermeersch}
\affiliation{Department of Mathematics,  KU Leuven Campus Kortrijk -- Kulak, Etienne Sabbelaan 53, 8500 Kortrijk, Belgium}
\email{ken.vandermeersch@kuleuven.be}

\begin{abstract}
\noindent This paper investigates a recent construction using bumpified Haar wavelets to demonstrate explicit violations of the Bell-Clauser-Horne-Shimony-Holt inequality within the vacuum state in Quantum Field Theory. The construction was tested for massless spinor fields in \((1+1)\)-dimensional Minkowski spacetime and is claimed to achieve violations arbitrarily close to an upper bound known as Tsirelson's bound. We show that this claim may be reduced to a mathematical conjecture involving the maximal eigenvalue of a sequence of symmetric matrices composed of integrals of Haar wavelet products. More precisely, the asymptotic eigenvalue of this sequence should approach $\pi$. We present a formal argument using a subclass of wavelets, allowing to reach $3.11052$. Although a complete proof remains elusive, we present further compelling numerical evidence to support it.
\end{abstract}

\maketitle

\section{Introduction}
Bell's inequalities, formulated in the 1960s by John Bell and collaborators \cite{Bell1964,Clauser1969}, have been central to testing the foundations of Quantum Mechanics.
A crucial aspect underlying these inequalities is relativistic causality, which ensures that measurements performed by space-like separated observers cannot influence each other. Given the role that relativistic causality plays, it seems natural to look at the Bell-CHSH inequality within the framework of Quantum Field Theory (QFT), although this introduces a number of challenges.

Using methods from Algebraic QFT, the pioneering works \cite{SummersI1987,SummersII1987,Summers1987} established an existence proof --- without an explicit construction --- showing that the Bell-CHSH inequality can be maximally violated at the level of free fields, arbitrarily close to Tsirelson's bound \(2\sqrt 2\) \cite{Cirelson80}.   For a recent review on the issue of Bell inequalities in the context of QFT, see \cite{Guimaraes:2024mmp}. The interest in Bell inequalities, or quantum entanglement in general, in the context of particle physics saw a recent increased interests, see for example \cite{ATLAS:2023fsd,Barr:2024djo,DeFabritiis:2024jfy,Fabbrichesi:2024rec,Ruzi:2024cbt,Altomonte:2024upf,Afik:2025ejh,Qi:2025onf,Florio:2025hoc}.

Recently, \cite{Dudal2023} introduced a general setup to \emph{explicitly} construct such violations in the vacuum state within the context of QFT. This approach was investigated in more detail for a free \((1+1)\)-dimensional massless spinor field, employing Haar wavelets, and yielded excellent numerical results. These results were based on two direct numerical optimizations, which provided strong evidence that violations can come very close to Tsirelson's bound. However, this method was computationally expensive.

Building on these observations, the present work reformulates the underlying hypothesis --- that maximal violations can be approached arbitrarily closely in the free massless (\(1+1\))-dimensional spinor case --- into a precise mathematical conjecture. This conjecture is expressed in terms of the asymptotic behavior of the largest eigenvalue of a sequence of block Toeplitz matrices constructed from Haar wavelet integrals (Conjecture \ref{conj:max}). We formally prove that if this conjecture holds, then the original hypothesis follows.

While a complete proof of Conjecture \ref{conj:max} remains open, testing the conjecture numerically is significantly more efficient than the previous direct optimization approach. We also provide compelling numerical evidence in its support, thereby strengthening the case for the existence of near-maximal violations in this QFT setting.


\section{Summary of previous results}\label{sec:recap}
The paper \cite{Dudal2023} roughly follows the following steps, summarizing its Sections V, VI and VII.   We refrain from giving too many details as the focus of this paper is squarely on the mathematical properties of the proposed construction. For more details, the interested reader is referred to, e.g.,~\cite{haag2012local}.
\begin{enumerate}
    \item A Quantum Field Theory (QFT) for a free, \((1+1)\)-dimensional, spinor field is introduced, quantized using the canonical anti-commutation relations.
    \item To properly define field operators in  Algebraic QFT, the fields are smeared using smooth test functions with compact support. This ensures the operators are well-behaved operators acting on a Hilbert space. More specifically, for each compactly supported, smooth, two-component spinor \textit{test function} \[h(t,x) = \big( h_1(t,x) , h_2(t,x) \big)^\mathsf T,\] we can associate a well-defined \textit{field operator} \(\psi(h)\). These smeared operators respect causality, with the anti-commutator \(\{\psi(h),\psi^\dagger(h')\}\) vanishing if the underlying test functions \(h\) and \(h'\) have space-like separated supports.
    \item The Bell operator is then defined as \(\mathcal A_h = \psi(h) + \psi^\dagger(h)\), which is Hermitian per construction. An inner product between test functions is defined as \(\big\langle h \mid h' \big\rangle = \big\langle 0 \mid \mathcal A_h \mathcal A_{h'} \mid 0 \big\rangle\), the vacuum expectation value of the product \(\mathcal A_h \mathcal A_{h'}\). In the special case \(h = h'\), we even have that \(\mathcal A_h^2 = \big\langle h \mid h \big\rangle \equiv \lVert h \rVert^2,\) showing that the Bell operators are dichotomic, meaning \(\mathcal A_h^2 = 1\), provided the underlying test functions are normalized with respect to this inner product.
    \item The Bell-CHSH correlator is then introduced as \[
    \big\langle \mathcal C \big\rangle = \big\langle 0 \mid i \left[(\mathcal A_f + \mathcal A_{f'})\mathcal A_g + (\mathcal A_f - \mathcal A_{f'})\mathcal A_{g'} \right]  \mid 0 \big\rangle,
    \] where \(f\) and \(f'\) are Alice's spinor test functions and \(g\) and \(g'\) are Bob's.  We remind here that Alice and Bob are agents, both performing measurements in their respective labs. The supports of \(f,f'\) and \(g,g'\) are located in Rindler's left and right wedges, respectively, ensuring relativistic causality.
       \item By expressing the Bell-CHSH correlator in terms of the inner products of these test functions, it becomes \[
    \big\langle \mathcal C \big\rangle = i \left( \big\langle f \mid g \big\rangle + \big\langle f \mid g' \big\rangle + \big\langle f' \mid g \big\rangle  - \big\langle f' \mid g' \big\rangle\right).\]
   \item  In practice, we explicitly implement relativistic causality by considering the hypersurface $t=0$, or better said, the temporal part of all test functions are considered to be of the type $\varepsilon^{-1} \beta_\pm(t\varepsilon^{-1})$, with e.g.~as choice of bump function $\beta_+(t)\propto e^{-\frac{1}{t(1-t)}}\mathds{1}_{[0,1]}$ and its reflected version $\beta_-(t)=\beta_+(-t)$.
  These bump functions have compact support on either $\mathbb{R}^+$ or $\mathbb{R}^-$ and should be properly normalized to $1$. For $\varepsilon\to 0^+$, it can then be shown in a distributional sense that this procedure corresponds to multiplication with a Dirac-$\delta(t)$ approached from the right or left. This is compatible with causality if we then simply restrict the supports of the remaining spatial part of the test functions, namely Alice's~$(f,f')$ to $x<0$, and Bob's~$(g,g')$ to $x>0$.

  \item  From now on, we will restrict ourselves to real-valued test functions. Focusing on the massless limit, these inner products admit then an analytical expression of the form \(\big\langle f \mid g \big\rangle = I_1 + I_2,\) where
   \begin{eqnarray}
 I_1 &=& \int \, \Big(f_1(x)^\ast g_1(x) +  f_2(x)^\ast g_2(x)\Big) \, \mathrm{d}x, \nonumber\\
I_2 &=&  -\frac{i}{\pi} \iint \left( \frac{1}{x-y} \right) \, \Big(f_1(x)^\ast g_1(y) -  f_2(x)^\ast g_2(y)\Big) \,\mathrm{d}x\mathrm{d}y.
   \end{eqnarray}
 Notice that, given the real-valuedness of the test functions, \[\big\langle f \mid f \big\rangle = \int \, \Big(f_1(x)^2  +  f_2(x)^2 \Big) \, \mathrm{d}x, \] and that \(\big\langle f \mid g \big\rangle = I_2\) whenever \(f\) and \(g\) are disjointly supported.
\end{enumerate}
This sets up the mathematical framework needed to investigate possible violations of Bell's inequalities: we need to find test functions \((f,f')\) for Alice and \((g,g')\) for Bob such that \[2 < \left\lvert  \big\langle \mathcal C  \big\rangle \right\rvert < 2 \sqrt{2}.\] In particular, we are interested in constructing test functions such that the inequality is maximally violated. That is, violations arbitrarily close to Tsirelson's bound: \[
\left\lvert  \big\langle \mathcal C  \big\rangle \right\rvert \longrightarrow 2 \sqrt 2.
\]

To achieve this, the following strategy is used: take \(\eta \in ( \sqrt 2 - 1 , 1)\)  as a given parameter (denoted \(\lambda\) in the previous paper \cite{Dudal2023}) and see if we can find a solution to the equations
\[
\big\langle {f} \mid {g} \big\rangle =  \big\langle {f}' \mid {g} \big\rangle =  \big\langle {f} \mid {g}' \big\rangle =  -\big\langle {f}' \mid {g}' \big\rangle = -i \frac{\sqrt2 \eta}{1+\eta^2}.
\]  If yes, it is easily checked that then \(\left\lvert  \big\langle \mathcal C  \big\rangle \right\rvert =  \frac{4\sqrt2 \eta}{1+\eta^2}\in \left(2, 2\sqrt 2\right)\) with \(\left\lvert  \big\langle \mathcal C  \big\rangle \right\rvert \longrightarrow 2 \sqrt 2\) as \(\eta \to 1.\) Concretely, this is done in two steps:
\begin{itemize}
    \item \textbf{Step 1.} We start by solving the system of equations for a set of preliminary test functions \((\widetilde{f},\widetilde{f}')\) and \((\widetilde{g},\widetilde{g}')\), using a finite series representation of \textit{Haar wavelets}.  This leads to piecewise constant functions, allowing the integrals in the calculations to be evaluated exactly.  See \cite[Section VIII]{Dudal2023}.
    \item \textbf{Step 2.} These Haar wavelets are then smoothed (bumpified) into continuous and differentiable versions. This bumpification ensures the final test functions are consistent with the smoothness requirements of Algebraic QFT.  See \cite[Section IX]{Dudal2023}.   In order to satisfy the locality requirement of Alice and Bob's test functions, for completeness we also implement a small finite translation to the test functions in this step.
\end{itemize}
In Section \ref{sec:probStat}, after briefly introducing the Haar wavelet basis and discussing a few relevant properties, we investigate Step 1, illustrating that the equations can be solved \textit{exactly}, provided the (finite) family of Haar wavelets used in the construction is taken sufficiently large --- this is the most challenging part of the process. Finally, in Section \ref{sec:bump}, we present a formal argument for why Step 2 works, provided we have already solved for the preliminary set of test functions from Step 1.

\section{An exact solution in terms of wavelets (Step 1)}\label{sec:probStat}
\subsection{Haar wavelets}\label{sec:haar}
Following \cite[Section VIII]{Dudal2023}, we employ Haar wavelets to construct the necessary test functions. Haar wavelets \cite{Lepik14} are a special type of Daubechies wavelets \cite{daub88}, used widely in signal processing and data compression \cite{kai94}.

\begin{definition}[Haar wavelets]
    The \textit{mother Haar wavelet} with support on the interval \([0,1)\) is defined by
    \[
    \psi(x) =
    \begin{cases}
        +1  &\text{ if } 0 \leq x < \dfrac 12 \\
        -1  &\text{ if } \dfrac 12 \leq x < 1\\
        0  &\text{ otherwise.}
    \end{cases}
    \]
  For every \(n,k\in \mathbb Z\), the Haar wavelet \(\psi_{n,k}\) is supported on the half-open interval \(I_{n,k} := [k 2^{-n}, (k+1)2^{-n})\) and is defined by \[
  \psi_{n,k}(x) =
        2^{n/2} \, \psi(2^n x - k).
\]\end{definition}
Thus, the \((n,k)\)th Haar wavelet \(\psi_{n,k}\) is a piecewise constant function, taking the value \(+2^{n/2}\) on the first half of the interval \(I_{n,k}\) and \(-2^{n/2}\) on the second half, being a rescaling of the mother Haar wavelet \(\psi\). See Figs.\ \ref{fig:haar} and \ref{fig:haarprop}. 

 The family of Haar wavelets \(\{\psi_{n,k}\}_{n,k \in \mathbb Z}\) is an \textit{orthonormal basis} for the space of square integrable functions \(L^2(\mathbb R)\); that is, the orthogonality condition \[
\int \psi_{n,k}(x) \psi_{m,\ell}(x) \, \mathrm{d}x = \delta_{n,k} \delta_{m,\ell} \quad \text{ for every \(n,k,m,\ell \in \mathbb Z\)}
\] is satisfied \textit{and} the linear span of \(\{\psi_{n,k}\}_{n,k \in \mathbb Z}\) is dense in \(L^2(\mathbb R)\), see \cite{Lepik14}.

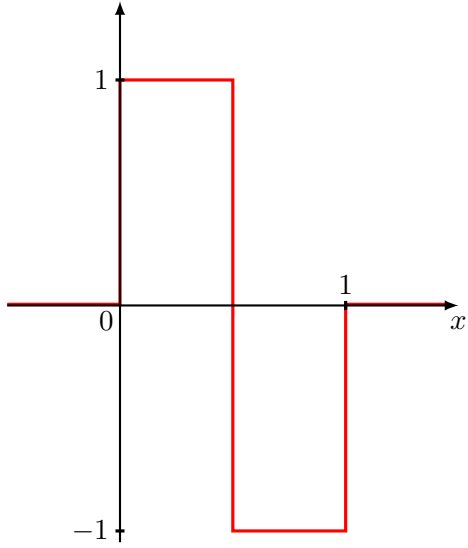
\begin{figure}
    \centering
\begin{tikzpicture}[scale=3, line width = .8pt]
		\def\xmin{-0.5}
		\def\xmax{1.5}
		\def\ymin{-1.05}
		\def\ymax{1.35}
        \def\eps{0.15}

        \draw[very thick, red] (\xmin, 0.005) -- (0,0.005) -- (0,1) -- (0.5,1) -- (0.5,-1) -- (1,-1) -- (1,0.005) -- (\xmax-0.05,0.005);

        \draw[very thick] (1,-0.02) --node[midway,above]{\(1\)} (1,0.02);

        \node[below left, inner sep = 2 pt] at (0,0) {\(0\)};
        \draw[very thick] (-0.02,1) --node[midway,left]{\(1\)} (0.02,1);
        \draw[very thick] (-0.02,-1) --node[midway,left]{\(-1\)} (0.02,-1);

        \draw[-latex]  (\xmin,0) -- (\xmax,0) node[below]{\(x\)};
		\draw[-latex]  (0,\ymin) -- (0,\ymax) ;

\end{tikzpicture}
\caption{The mother Haar wavelet \(\psi = \psi_{0,0}.\)} \label{fig:haar}
\end{figure}

The following symmetry property is easy to prove; see Fig.\ \ref{fig:haarprop} for an intuitive visualization.
\begin{lemma}\label{prop:haar} For all \(n,k \in \mathbb Z\) we have \[\displaystyle {\psi_{n,k}(-x) = -\psi_{n,-k-1}(x)} \quad \text{ almost everywhere.}\]
\end{lemma}
With ``almost everywhere'', we mean it holds everywhere except on a subset of $\mathbb{R}$ of measure zero.

\begin{figure}
   \centering
\begin{tikzpicture}[scale=3.2, line width = .8pt]
		\def\xmin{-2}
		\def\xmax{2}
		\def\ymin{-1.05}
		\def\ymax{1.35}
        \def\eps{0.15}

        \draw[very thick, red]  (1,0.005) -- (1,1) -- (1.35,1) node[right]{\(\psi_{n,k}(x)\)} -- (1.35,-1) -- (1.7,-1) -- (1.7,0.005) ;

         \draw[very thick, blue]  (-1.7,0.005) -- (-1.7,1) -- (-1.35,1) node[right]{\(\psi_{n,-k-1}(x)\)} -- (-1.35,-1) -- (-1,-1) -- (-1,0.005) ;

        \draw[ultra thick] (1,-0.02) --node[midway,below]{\small\(k2^{-n}\)} (1,0.02);
        \draw[ultra thick] (1.7,-0.02) --node[midway,above]{\small\((k+1)2^{-n}\)} (1.7,0.02);
        \draw[ultra thick] (-1,-0.02) --node[midway,above]{\small\(-k2^{-n}\)} (-1,0.02);
        \draw[ultra thick] (-1.7,-0.02) --node[midway,below]{\small\(\textstyle {(-k-1)2^{-n}}\)} (-1.7,0.02);


        \node[below left, inner sep = 2 pt] at (0,0) {\(0\)};
        \draw[very thick] (-0.02,1) --node[midway,left]{\(2^{n/2}\)} (0.02,1);
        \draw[very thick] (-0.02,-1) --node[midway,left]{\(-2^{n/2}\)} (0.02,-1);

        \draw[-latex]  (\xmin,0) -- (\xmax,0) node[below]{\(x\)};
		\draw[-latex]  (0,\ymin) -- (0,\ymax) ;

\end{tikzpicture}
\caption{The almost everywhere\ equality \(\psi_{n,k}(-x) = - \psi_{n,-k-1}(x)\), visually.} \label{fig:haarprop}
\end{figure}
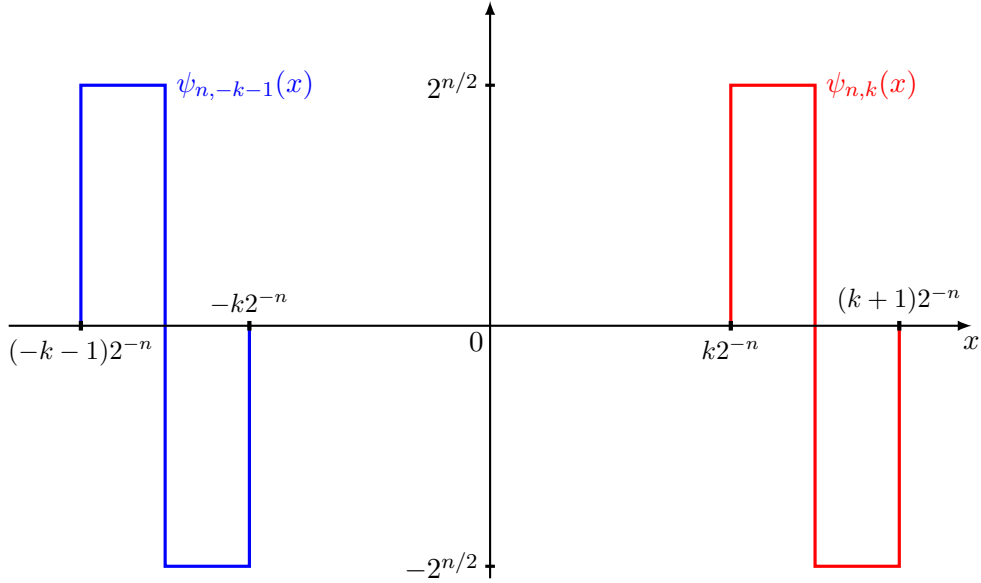

In this work, we will study specific integrals containing these Haar wavelets, of the following general form.
\begin{notation}\label{eq:A}
     For all \(n,k,m,\ell \in \mathbb Z\) with \(k,\ell\geq 1\), we define \[
     A_{(n,k),(m,\ell)} = -\iint \left( \frac{1}{x+y} \right) \, \psi_{n,-k}(x)\psi_{m,-\ell}(y) \,\mathrm{d}x\mathrm{d}y.
     \]
\end{notation}
\begin{lemma}\label{prop:A}
     For all \(n,k,m,\ell \in \mathbb Z\) with \(k,\ell\geq 1\), we have \(A_{(n,k),(m,\ell)} > 0.\)
\end{lemma}
\begin{proof}
We start by noting that the absolute value of the integrand, namely \(|x+y|^{-1}\), is integrable on the domain \(I_{n,-k}\times I_{m,-\ell}\).  After a change of variables, it suffices to prove that \(|x+y|^{-1}\) is integrable on \([0,M]\times[0,M]\) for every \(M > 0\). For completeness, we provide the details of this computation in Appendix~\ref{app:A}.

    Hence, by Fubini's theorem, we may first integrate with respect to \(x\):
\begin{align*}
    A_{(n,k),(m,\ell)} &= -\int \psi_{m,-\ell}(y) \left( \int \frac{\psi_{n,-k}(x)}{x+y} \, \mathrm{d}x \right) \mathrm{d}y.
\end{align*} Then, for each \(y < 0\), the inner integral may be computed as
\begin{align*}
    I(y) := \int \frac{\psi_{n,-k}(x)}{x+y} \, \mathrm{d}x
    ={}& \int_{-k2^{-n}}^{(-k+\frac12)2^{-n}} \frac{+2^{n/2}}{x+y} \, \mathrm{d}x + \int_{(-k+\frac12)2^{-n}}^{(-k+1)2^{-n}} \frac{-2^{n/2}}{x+y} \, \mathrm{d}x \\
    ={}& 2^{n/2}\ln\left(\frac{(-k+\frac12)2^{-n} + y}{-k2^{-n}+y} \right) - 2^{n/2}\ln\left(\frac{(-k+1)2^{-n}+y}{(-k+\frac12)2^{-n} + y} \right).
\end{align*}
Now remark that the function \(I \colon (-\infty,0) \longrightarrow \mathbb R\) is strictly increasing. Indeed, we may write \(I(y) = 2^{n/2}J(2^n y-k)\), where \(J \colon (-\infty, -k ) \subset (-\infty, -1) \longrightarrow \mathbb R \) is given by  \[
J(x) = \ln \left( \frac{\frac12 +x}{x} \right) - \ln \left( \frac{{1} +x}{\frac12 +x} \right)= \ln \left( 1+ \frac{1}{4x(x+1)} \right),\] and where \(\dfrac{1}{4x(x+1)}\) is strictly increasing for \(x < -1\).  The graph of \(J\) is depicted in Fig.\ \ref{fig:J}.  We conclude that
\begin{align*}
    A_{(n,k),(m,\ell)} &= -\int \psi_{m,-\ell}(y) I(y) \, \mathrm{d}y.  \\
    &= - 2^{m/2} \int_{-\ell2^{-m}}^{(-\ell+\frac12)2^{-m}} I(y)  \, \mathrm{d}y + 2^{m/2}\int_{(-\ell+\frac12)2^{-m}}^{(-\ell+1)2^{-m}} I(y)  \, \mathrm{d}y >0.
\end{align*}
\end{proof}
\begin{figure}
    \centering
    \includegraphics[width=0.6\linewidth]{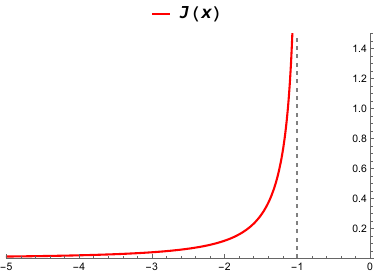}
    \caption{The graph of the function \(J\) from the proof of Lemma \protect{\ref{prop:A}}.}
    \label{fig:J}
\end{figure}

\begin{remark} \label{rem:J}
    The function \(J\) defined in the proof of Lemma \ref{prop:A} will return in later parts of this work. As stated in the proof, it is strictly increasing. It can also be verified that it is a strictly positive function, and that it defines a bijection \((-\infty, -1) \longrightarrow (0,\infty)\).
    \end{remark}


\subsection{Problem statement}

Resuming from Section \ref{sec:recap}, we employ the Haar wavelets introduced in Subsection \ref{sec:haar} to solve the following system of equations:
\begin{equation} \label{eq:28}
    \begin{array}{l}
  \big\langle \widetilde{f} \mid \widetilde{f} \big\rangle = \big\langle \widetilde{f}' \mid \widetilde{f}' \big\rangle = \big\langle \widetilde{g} \mid \widetilde{g} \big\rangle = \big\langle \widetilde{g}' \mid \widetilde{g}' \big\rangle = 1 \\
    \big\langle \widetilde{f} \mid \widetilde{g} \big\rangle =  \big\langle \widetilde{f}' \mid \widetilde{g} \big\rangle =  \big\langle \widetilde{f} \mid \widetilde{g}' \big\rangle =  -\big\langle \widetilde{f}' \mid \widetilde{g}' \big\rangle = -i \dfrac{\sqrt2 \eta}{1+\eta^2}.
    \end{array}
\end{equation}
 Here,
\(\eta \in (\sqrt2 - 1 , 1)\) is a given parameter. (This corresponds to \cite[Eq.\ (28)]{Dudal2023}, which the original authors solve numerically.) These initial test functions are expanded as finite linear combinations of the Haar wavelets:
\begin{equation}
   \begin{array}{c}
\displaystyle{ \widetilde f_j  := \sum_{n = N_0}^{N_1} \sum_{k = -K}^{-1} f_{j}(n,k) \, \psi_{n,k}, \qquad
    \widetilde g_j := \sum_{m = N_0}^{N_1} \sum_{\ell = 0}^{K-1} g_{j}(m,\ell) \, \psi_{m,\ell},} \\
\displaystyle{ \widetilde f_j^{\prime}  := \sum_{n = N_0}^{N_1} \sum_{k = -K}^{-1} f_{j}^{\prime}(n,k) \, \psi_{n,k}, \qquad
    \widetilde g_j^{\prime} := \sum_{m = N_0}^{N_1} \sum_{\ell = 0}^{K-1} g_{j}^{\prime}(m,
    \ell) \, \psi_{m,\ell},}
   \end{array} \label{eq:expansion}
\end{equation}
for \(j\in \{1,2\}\), so that, after subtituting these expansions, we only need to solve the system of equations \eqref{eq:28} for the finite set of wavelet coefficients \(f_{j}(n,k)\), \(g_{j}(m,\ell),f_{j}^{\prime}(n,k)\), \(g_{j}^{\prime}(m,\ell)\).

The parameters \(N_0 < N_1\) and \(K > 1\) set the resolution; increasing the resolution results in more Haar wavelets being used and the system of equations \eqref{eq:28} thus becoming easier to solve. The idea will be that, for a given \(\eta\in (\sqrt2-1 , 1)\) \textit{arbitrarily close} to \(1\), we can find a resolution \textit{sufficiently fine} such that the system of equations \eqref{eq:28} has an \textit{exact} solution for the wavelet coefficients. We note that the causality condition is implemented by only using Haar wavelets supported on the  region  \([-2^{-N_0}K, 0)\) to represent Bob's test functions \(f, f'\) and only using Haar wavelets supported on the  region  \([0,2^{-N_0}K)\) to represent Alice's test functions \(g,g'\), see \cite[Section VIII]{Dudal2023}.  Technically speaking, according to \cite{SummersII1987}, the two considered regions of Alice on one hand and Bob on the other hand, are supposed to be complementary wedges, that is, each other's causal complement \emph{up to taking the closure}, see \cite{Bisognano:1975ih,Bisognano:1976za}. As such, the ``touching point'' at $x=0$ is allowed, cf.~the expansions \eqref{eq:expansion}. As noted in \cite[footnote 21]{Bisognano:1976za}: ``how the boundaries are handled is then only a question of mathematical convenience.''. Also note that, eventually, the bumpified versions of the test functions (cf.~section.~\ref{sec:bump}) will be vanishing exponentially fast when approaching the touching point, from both sides, so we expect that this point will be of no numerical relevance anyhow.  Though, for completeness, we will nonetheless provide a formal argument as to why ``the touching point'' does not matter in the final test functions by implementing a small translation, see Subsection~\ref{sub:bumpTest}.

The following conjecture states that this procedure always works and we spend the rest of this section providing evidence for it.
\begin{conjecture}\label{conj:soluble}
Given any \(\eta \in (\sqrt2 - 1 , 1)\), a resolution \(\{N_0,N_1,K\}\) can be found such that the system of equations \eqref{eq:28} admits a solution of the form \eqref{eq:expansion}.
\end{conjecture}
Substituting the expansions \eqref{eq:expansion} into the considered norms and inner products under consideration \eqref{eq:28}, yields expressions of the form
\begin{equation}
     \big\langle \widetilde{f} \mid \widetilde{f} \big\rangle = \sum_{n,k} \left( f_1(n,k)^2 + f_2(n,k)^2 \right),
\end{equation} and
\begin{align}
   \big\langle \widetilde{f} \mid \widetilde{g} \big\rangle = \sum_{n,k,m,\ell} & \Big(f_1(n,k)g_1(m,\ell) - f_2(n,k)g_2(m,\ell) \Big) \notag \\
  &\times\left(-\frac{i}{\pi}\iint \left( \frac{1}{x-y} \right) \, \psi_{n,k}(x)\psi_{m,\ell}(y) \,\mathrm{d}x\mathrm{d}y \right), \label{eq:innerprod}
\end{align}
where both sums are finite and depend on the choice of resolution. Of course, similar expressions hold for the other norms and inner products entering the system of equations \eqref{eq:28}, see also \cite[Eqs.\ (26), (27)]{Dudal2023}.

\subsection{Numerical solution}
The authors of \cite{Dudal2023} then proceeded to numerically solve the system of equations \eqref{eq:28} for the unknown wavelet coefficients, considering two cases: \(\eta = 0.7\) and \(\eta = 0.99\). This yielded corresponding correlations \( \left\lvert  \big\langle \mathcal C  \big\rangle \right\rvert \approx 2.66\) and \(\left\lvert  \big\langle \mathcal C  \big\rangle \right\rvert  \approx 2.82\), both already quite close to maximal violation \(2\sqrt 2 \approx 2.83\). Bumpified versions of these test functions were subsequently found with numerically indistinguishable violations. The final test functions corresponding to \(\eta = 0.99\) are depicted in Fig.\ \ref{fig:antiSym}, a copy of \cite[Fig.\ 2]{Dudal2023}.

The graphs of these functions seem to exhibit a number of (anti)symmetry relations which we intend to exploit to simplify the system of equations \eqref{eq:28}. Doing this, we will convert Conjecture \ref{conj:soluble} into a more direct and better testable mathematical statement, Conjecture \ref{conj:max}, concerning the maximal eigenvalue of a certain matrix.

\begin{figure}
    \centering
    \includegraphics[width=0.65\linewidth]{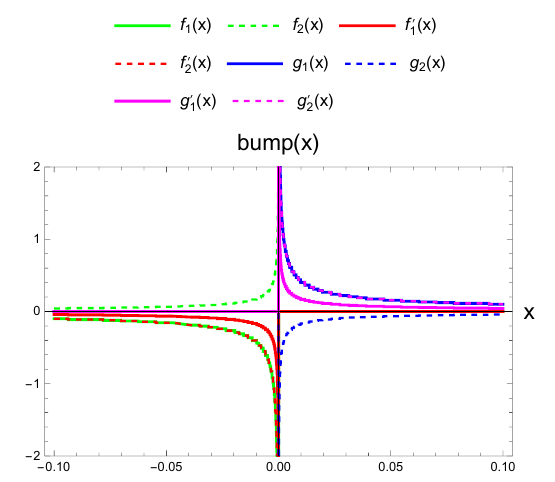}
    \caption{Test function components for \(\eta=0.99\) where the resolution was set to \(N_0 = -10\), \(N_1 = 120\), \(K = 5\). \protect{\cite[Fig.\ 2]{Dudal2023}} }
    \label{fig:antiSym}
\end{figure}

\subsection{Exploiting the (anti)symmetry relations}
Consider the expressions for the inner products, Eqs.\ \eqref{eq:innerprod}.  Applying a change of variables \(v = -y\) and using the almost everywhere\ equality of Lemma \ref{prop:haar}, we may rewrite these as follows:
\begin{align*}
     \iint \left( \frac{1}{x-y} \right) \, \psi_{n,k}(x)\psi_{m,\ell}(y) \,\mathrm{d}x\mathrm{d}y &=   \iint \left( \frac{1}{x+v} \right) \, \psi_{n,k}(x)\psi_{m,\ell}(-v) \,\mathrm{d}x\mathrm{d}v  \\
    &=  - \iint \left( \frac{1}{x+v} \right) \, \psi_{n,k}(x)\psi_{m,-\ell-1}(v) \,\mathrm{d}x\mathrm{d}v  \\
    &\equiv A_{(n,-k),(m,\ell+1)}.  \tag{By Eq.\ \eqref{eq:A}}
\end{align*}
Hence, after reindexing the summations in Eq.\ \eqref{eq:innerprod} by \(k \leftarrow -k\) and \(\ell \leftarrow \ell+1\) (so that \(k\) and \(\ell\) both run over the range \(\{1,\ldots,K\}\)), we deduce
\[
\big\langle \widetilde{f} \mid \widetilde{g} \big\rangle = - \frac{i}{\pi} \! \sum_{n,k,m,\ell}  A_{(n,k),(m,\ell)} \Big(f_1(n,-k)g_1(m,\ell-1) - f_2(n,-k)g_2(m,\ell-1) \Big),
\] with similar expressions for the other inner products.

Now, we proceed to simplify the problem statement \eqref{eq:28} significantly, reducing the number of coefficients that we need to solve for by imposing three additional constraints, exploiting specific (anti)symmetries which may be observed in Fig.\ \ref{fig:antiSym}.

\begin{enumerate}
    \item First, we note from Fig.\ \ref{fig:antiSym} that \(f_2' = f_1\) and \(f_1' = -f_2\) for Alice's spinor test functions, with similar expressions holding for Bob's. Hence, we assume the solution satisfies
\begin{align*}
 f_2'(n,k) = f_1(n,k), \qquad f_1'(n,k) = -f_2(n,k), \\
 g_2'(n,k) = g_1(n,k), \qquad g_1'(n,k) = - g_2(n,k).
\end{align*}
Then \(\big\langle \widetilde{f} \mid \widetilde{g} \big\rangle = - \big\langle \widetilde{f}' \mid \widetilde{g}' \big\rangle\) and \(\big\langle \widetilde{f}' \mid \widetilde{g} \big\rangle = \big\langle \widetilde{f} \mid \widetilde{g}' \big\rangle\) in \eqref{eq:28} hold automatically. For example, the first equality follows from
\begin{align*}
   f_1'(n,-k)g_1'(m,\ell-1) - &f_2'(n,-k)g_2'(m,\ell-1)  \\
   &=   f_2(n,-k)g_2(m,\ell-1) - f_1(n,-k)g_1(m,\ell-1) \\
   &= -\Big( f_1(n,-k)g_1(m,\ell-1) - f_2(n,-k)g_2(m,\ell-1)\Big).
\end{align*}
\item Another relation which may be observed from Fig.\ \ref{fig:antiSym} is that \( f(-x) = -  g(x)\). This relation may be enforced by imposing
\[
g_1(n,k) = f_1(n,-k-1) \quad \text{ and } \quad g_2(n,k) = f_2(n,-k-1).
\]
Indeed, by Lemma \ref{prop:haar}, we deduce that
\begin{align*}
    f_j(-x) &= \sum_{n = N_0}^{N_1} \sum_{k = 1}^{K} f_{j}(n,-k) \, \psi_{n,-k}(-x) \\
    &= -\sum_{n = N_0}^{N_1} \sum_{k = 1}^{K} g_j(n,k-1) \, \psi_{n,k-1}(x) \\
    &= -g_j(x) \qquad \text{almost everywhere.}
\end{align*}
\item Lastly, again based on Fig.\ \ref{fig:antiSym}, it is plausible that \(f_2 = -c f_1\) for some constant \(0.25 < c< 0.5\). Hence, we also assume \(f_2(n,k) = -c f_1(n,k),\) where \(c\) is a constant yet to be determined.
\end{enumerate}
Under these three additional assumptions, Conjecture \ref{conj:soluble} simplifies to finding a resolution \(\{N_0,N_1,K\}\) such that the following system of equations:
\begin{align*}
 \sum_{n,k,m,\ell} \! A_{(n,k),(m,\ell)} \Big(  f_1(n,-k) f_1(m,-\ell) - c^2 f_1(n,-k) f_1(m,-\ell) \Big) &=  \frac{\pi \sqrt{2} \eta}{1+\eta^2}  \\
 \sum_{n,k,m,\ell} \! A_{(n,k),(m,\ell)} \Big( -cf_1(n,-k) f_1(m,-\ell) -c f_1(n,-k) f_1(m,-\ell)  \Big) &= - \frac{\pi \sqrt{2} \eta}{1+\eta^2},
\end{align*}
admits an exact solution satisfying the normality condition \[{\sum_{n,k} f_1(n,-k)^2 + c^2 f_1(n,-k)^2 = 1}.\]

For brevity, we denote \(x_{n,k} := f_1(n,-k)\) for all \(N_0 \leq n \leq N_1\) and \(1 \leq k \leq K\). We may put these unknowns in a vector \( x\), indexed by these \((n,k)\), and ordered as follows:
\[
{x} =
\begin{bmatrix}
    x_{N_0,1}\! & \cdots &  \!x_{N_0,K} \ \Big| \   x_{N_0+1,1} \quad \quad  \cdots \quad \quad   x_{N_1-1,K} \ \Big|\ x_{N_1,1}\! & \cdots & \!x_{N_1,K}
\end{bmatrix}^{\mathsf T} \in \mathbb R^{d},
\] where \(d = (N_1 - N_0 + 1)\times K\).  This compactifies the problem statement into
\begin{align}
 (1-c^2) \sum_{n,k,m,\ell}A_{(n,k),(m,\ell)} \, x_{n,k} \, x_{m,\ell} &=  \frac{\pi \sqrt{2} \eta}{1+\eta^2} \label{eq:28'1} \\
  -2c \sum_{n,k,m,\ell}A_{(n,k),(m,\ell)} \, x_{n,k}\, x_{m,\ell} &= - \frac{\pi \sqrt{2} \eta}{1+\eta^2}, \label{eq:28'2}
\end{align}
and the normality condition on \(f\) becomes the following normality condition on \(x\): \[\lVert x \rVert^2 = 1/(1+c^2 ).\]

A natural first step towards a solution to this system of equations is to add both equations \eqref{eq:28'1} + \eqref{eq:28'2}, obtaining
\[
(1-2c-c^2)  \sum_{n,k,\ell,m} A_{(n,k),(m,\ell)} \, x_{n,k} \, x_{m,\ell}  = 0.
\] This equation can be satisfied by simply taking \(c\) as the positive root of the quadratic polynomial \(1-2x-x^2\), namely \(c = \sqrt{2} -1 \approx 0.4142.\)

\begin{figure}
    \centering
 \begin{subfigure}[b]{0.45\textwidth}
    \centering
\includegraphics[width=1\textwidth]{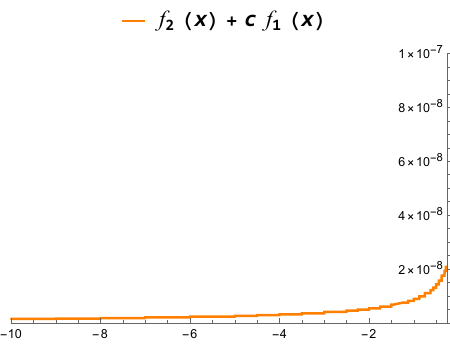}
\end{subfigure}
     \hfill
\begin{subfigure}[b]{0.45\textwidth}
    \centering
\includegraphics[width=1\textwidth]{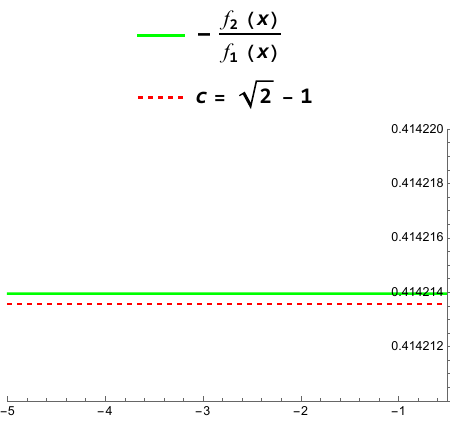}
\end{subfigure}
\caption{Numerical evidence that the test function components found in \protect{\cite{Dudal2023}} satisfy \(f_2 = -cf_1\).}
\label{fig:c-value}
\end{figure}

\begin{remark}
    We note the excellent agreement between this precise value for \(c\) and the one obtained numerically by using the test functions corresponding to \(\eta=0.99\) found in \cite{Dudal2023}. The components of the final test function \(f = (f_1,f_2)\) should (approximately) satisfy \(f_2 = -cf_1\), where \(c = \sqrt 2 - 1\). In Fig.\ \ref{fig:c-value} we have depicted the graphs of both \(f_2 + c f_1 \approx 0\) and \(- f_2 / f_1 \approx c\).
\end{remark}

Next, as we require the unknown vector  \(x\) to satisfy \(\lVert x \rVert^2 = 1/(1+c^2 )\), it makes sense to rescale the set of unknowns by looking for a vector of the form \[y := \sqrt{1+c^2 } \; x, \] which is normalized.  Accordingly, all we still have to do now is solve
\[
\sum_{(n,k)} \sum_{(m,\ell)} A_{(n,k),(m,\ell)} \, y_{n,k} \, y_{m,\ell} =  \frac{{1+c^2}}{1-c^2} \cdot  \frac{\pi \sqrt{2} \eta}{1+\eta^2} = \frac{2\pi \eta}{1+\eta^2},\]
such that \(\lVert y \rVert^2 = 1\).

\subsection{The matrix \texorpdfstring{\(A\)}{A}}
Analogously to the definition of the vectors \(x\) and \(y\), we define the \(d \times d\) symmetric matrix \[A = \left[ A_{(n,k),(m,\ell)} \right]_{(n,k),(m,\ell)},\] with entries given by Eq.\ \eqref{eq:A},
further compactifying the problem statement.
In the previous subsection, we have proved that Conjecture \ref{conj:soluble} may be reduced to the following statement:
\begin{quote}
    Find a solution \(y \in \mathbb R^{d}\) to
\begin{equation}
   y^{\mathsf{T}} A y = \frac{2\pi \eta}{1+\eta^2} \qquad \text{such that} \qquad  \lVert y \rVert^2 = 1. \label{eq:problem}
\end{equation}
\end{quote}
The existence of such a vector depends on the eigenvalues of \(A\), which, in turn, depend on the parameters \(N_0\), \(N_1\) and \(K\). By the spectral theorem, \(A\) admits an eigenvalue decomposition \[A = Q \begin{bmatrix}
    \lambda_1 & & & \\
    & \lambda_2 & & \\
    & & \ddots & \\
    & & & \lambda_d
\end{bmatrix} Q^\mathsf T,\] where \(Q\) is an orthogonal matrix and \(\lambda_{\mathrm{min}} = \lambda_d \leq \ldots \leq \lambda_2 \leq  \lambda_1 = \lambda_{\mathrm{max}}\).
Applying the Courant-Fischer min-max principle, the existence of a solution to problem statement \eqref{eq:problem} can be characterized in terms of the extremal eigenvalues \(\lambda_{\mathrm{min}} \) and \(\lambda_{\mathrm{max}} \).
\begin{lemma}\label{prop:probstat}
    The problem statement \eqref{eq:problem} has a solution if and only if \[\lambda_{\min} \leq \frac{2\pi \eta}{1+\eta^2} \leq \lambda_\mathrm{max}.\]
\end{lemma}

Interestingly, the matrix \(A\) depends only on the difference \(N:= N_1-N_0\).
This follows from the observation that \[
A_{(n+1,k),(m+1,\ell)} = A_{(n,k),(m,\ell)}  \quad \text{ for all \(n,m \in \mathbb Z\) and \(k,\ell \in \mathbb Z_{\geq 1}\).}
\]
As a result, to study the spectrum of \(A\), we can set \(N_0 = 0\) and \(N_1 = N\) without loss of generality. To emphasize this dependence, we also write \(A(N,K)\) instead of \(A\). As an example, Table \ref{tab:N2K3} shows the matrix \(A(2,3)\), i.e., the matrix \(A\) for \(N = 2\) and \(K = 3\).

\begin{table}
{ \scriptsize \[
\left[
\begin{array}{ccc|ccc|ccc}
0.340& 0.0179& 0.00490& 0.321& 0.0284& 0.00919& 0.274& 0.0360&  0.0139, \\
  0.0179& 0.00490& 0.00202& 0.00919& 0.00414& 0.00222&  0.00405& 0.00256& 0.00173, \\
  0.00490& 0.00202& 0.00102& 0.00222&  0.00133& 0.000857& 0.000899& 0.000679& 0.000526, \\ \hline
  0.321& 0.00919& 0.00222& 0.340& 0.0179& 0.00490& 0.321& 0.0284& 0.00919, \\
  0.0284&  0.00414& 0.00133& 0.0179& 0.00490& 0.00202& 0.00919& 0.00414&  0.00222, \\
  0.00919& 0.00222& 0.000857& 0.00490& 0.00202& 0.00102&  0.00222& 0.00133& 0.000857, \\ \hline
  0.274& 0.00405& 0.000899& 0.321&  0.00919& 0.00222& 0.340& 0.0179& 0.00490, \\
  0.0360& 0.00256&   0.000679& 0.0284& 0.00414& 0.00133& 0.0179& 0.00490&  0.00202, \\
  0.0139& 0.00173& 0.000526& 0.00919& 0.00222& 0.000857&  0.00490& 0.00202& 0.00102
\end{array}
\right] \]}
    \caption{The matrix \(A(2,3)\). Notice the \(K \times K\) block structure.}
    \label{tab:N2K3}
\end{table}
We are particularly interested in the behavior as \(\eta \to 1\), where \(
\dfrac{2\pi \eta}{1+\eta^2} \longrightarrow \pi.\) This suggests that we should be able to make the largest eigenvalue \(\lambda_{\max}\big(A(N,K)\big)= \lVert A(N,K) \rVert_2\) approach \(\pi\) by taking the resolution \(N\), \(K\) sufficiently large.  Numeric results support this expectation: this can be seen in Fig.\ \ref{fig:ANK}, where we have depicted discrete plots of \(\lVert A(N,K) \rVert_2\) for \(2 \leq N \leq 50\), once for \(K = 2\) and once for \(K=5\).

\begin{figure}
    \centering
 \begin{subfigure}[c]{0.7\textwidth}
    \centering
\includegraphics[width=1\textwidth]{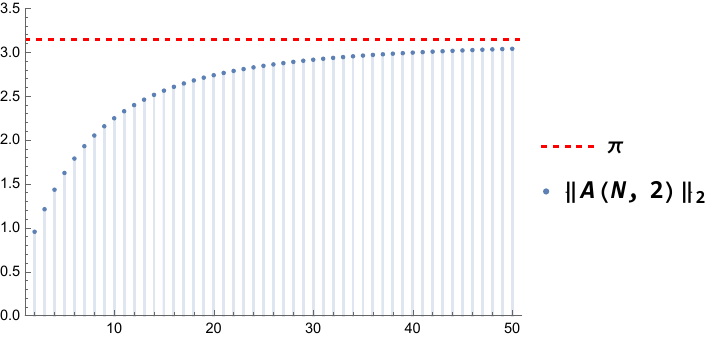}
\caption{\(\lambda_{\max}\big(A(N,2)\big)\) as a function of \(N\)}
\end{subfigure} \bigskip \bigskip

\begin{subfigure}[c]{0.7\textwidth}
    \centering
\includegraphics[width=1\textwidth]{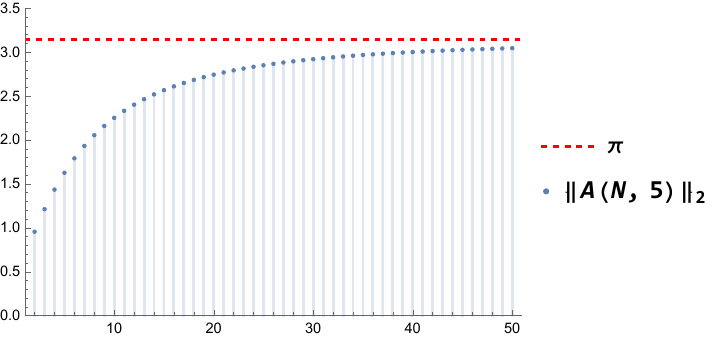}
\caption{\(\lambda_{\max}\big(A(N,5)\big)\) as a function of \(N\)}
\end{subfigure} \bigskip \bigskip

\begin{subfigure}[c]{0.7\textwidth}
    \centering
\includegraphics[width=1\textwidth]{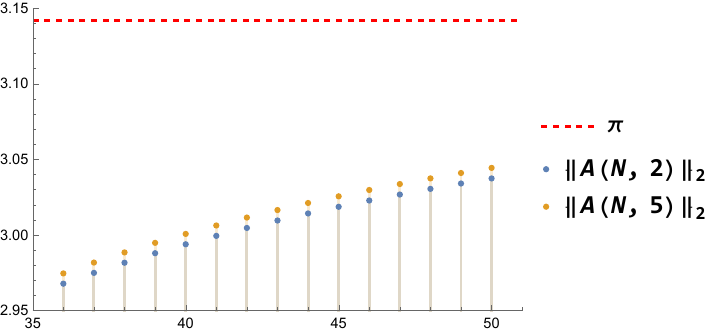}
\caption{Comparing \(\lambda_{\max}\big(A(N,2)\big)\) and \(\lambda_{\max}\big(A(N,5)\big)\)}
\end{subfigure}

\caption{Discrete plots for \(\lambda_{\max}\big(A(N,2)\big)\) and \(\lambda_{\max}\big(A(N,5)\big)\), together with a comparison.} \label{fig:ANK}
\end{figure}

\begin{remark}[Minimal eigenvalue]
   As we are particularly interested in the behavior as \(\eta \to 1\),  the \textbf{minimal} eigenvalue of the matrix \(A(N,K)\) plays only a minor role. Indeed, for any \(N, K \geq 1\), we can guarantee that \[
    \lambda_{\min}\big( A(N,K) \big) \leq A_{(0,1),(0,1)} = \ln\left( \frac{1024}{729} \right) \approx 0.340.
    \]  This inequality holds because \(A_{(0,1),(0,1)}\) is a diagonal entry of the symmetric matrix \(A(N,K)\), and thus must be greater than or equal to the minimal eigenvalue.\footnote{To see why this is the case, consider the inequality \[ \forall {y \in \mathbb R^d_{\neq 0}}\colon \quad \lambda_{\min}\big( A(N,K) \big) \leq \frac{y^\mathsf T A(N,K) y}{y^\mathsf T y} ,
    \] (another application of Courant-Fischer) and plug in the first basis vector \(y = \begin{bmatrix}
        1 & 0 & \cdots &0
    \end{bmatrix}^\mathsf T\).}  Hence, the lower bound in Lemma \ref{prop:probstat} is satisfied by default, since \[
    \min_{\eta\in[\sqrt 2 - 1, 1]}\frac{2\pi \eta}{1+\eta^2} = \frac{2\pi (\sqrt 2 - 1)}{1+(\sqrt 2 - 1)^2} \approx 2.22.
    \]
\end{remark}

Let us formally state our conjecture concerning the maximal eigenvalue of the matrix \(A(N,K)\).

\begin{conjecture}\label{conj:max}
 Given any \(\delta > 0\), we can find sufficiently large \(N, K \geq 1\) such that  \[\pi - \delta < \lambda_{\max}\big(A(N,K)\big).\] In particular, this implies that given any \(\eta \in (\sqrt 2 - 1, 1)\), we can find sufficiently large \(N, K \geq 1\) such that  \[\frac{2\pi \eta}{1+\eta^2} \leq \lambda_{\max}\big(A(N,K)\big).\]
\end{conjecture}
We have proved that Conjecture \ref{conj:max} is sufficient to prove our earlier hypothesis Conjecture \ref{conj:soluble}.
\begin{proposition}
    Conjecture \ref{conj:max} \(\implies\) Conjecture \ref{conj:soluble}.
\end{proposition}

\begin{remark}\label{rem:underlyingPhysics}
    In fact, Conjecture \ref{conj:max} is equivalent to \[\lim_{N,K\to \infty}\lambda_{\max}\big(A(N,K)\big) = \pi.\] This is because:
    \begin{itemize}
\item Once we have found \(N, K \geq 1\) such that  \(\pi - \delta < \lambda_{\max}\big(A(N,K)\big)\), then also  \(\pi - \delta <  \lambda_{\max}\big(A(N',K')\big)\) for each  \(N' \geq N\) and each \(K' \geq K\). This is because \(A(N,K)\)  is a submatrix of \(A(N',K')\).
\item \(\lambda_{\max}\big(A(N,K)\big) \leq \pi\), as otherwise a solution may be found which violates the Tsirelson bound in free QFT (\cite{SummersI1987,SummersII1987,Summers1987,Cirelson80}).
    \end{itemize}
\end{remark}

\noindent Returning to the example in Table \ref{tab:N2K3} (\(N=2\), \(K=3\)), by subdividing the matrix into \(3 \times 3\) blocks, we notice that \(A(2,3)\) is \textit{block Toeplitz}. This means that the matrix \(A(2,3)\) is of the following form:
\[
A = \begin{bmatrix}
A_{0}&A_{1}^\mathsf T &A_{2}^\mathsf T\\
A_1 & A_{0}&A_{1}^\mathsf T \\
A_{2}& A_{1}&A_{0}
\end{bmatrix}. \]This is the case for general \(N\) and \(K\).
\begin{proposition}\label{prop:blockToep}
    The symmetric matrix \(A(N,K)\) is block Toeplitz: visually, we have
 \[
A = \begin{bmatrix}
A_0           & A_{1}^\mathsf {T} & A_{2}^\mathsf{T} & \cdots  & A_{N}^\mathsf{T} \\
A_1           & A_0              & A_{1}^\mathsf{T} & \cdots  & A_{N-1}^\mathsf{T} \\
A_2           & A_1              & A_0              & \cdots  & A_{N-2}^\mathsf{T} \\
\vdots        & \vdots           & \vdots           & \ddots & \vdots            \\
A_{N}         & A_{N-1}          & A_{N-2}          & \cdots & A_0
\end{bmatrix},
\]    where the individual blocks \(A_0(K)\), \(\ldots\), \(A_N(K)\) are of size \(K \times K\) and of the form \[
A_n(K) = \Big[A_{(n,k),{(0,\ell)}} \Big]_{\substack{1 \leq k \leq K \\ 1 \leq \ell \leq K}}.
\] 
\end{proposition}
\begin{proof}[Proof of Proposition \ref{prop:blockToep}]
    We already know that the matrix \(A(N,K)\) is symmetric, so we only have to check that the blocks satisfy the Toeplitz structure. By the definition of \(A(N,K)\), it follows that the blocks are indexed by \((n,m)\): for example, for  \(N = 3\), the block indices are given by:
    \[
\begin{bmatrix}
(0,0) & (0,1) & (0,2) & (0,3) \\
(1,0) & (1,1) &  (1,2) &  (1,3) \\
(2,0)&(2,1)&(2,2)&(2,3)\\
(3,0)&(3,1)&(3,2)&(3,3)
\end{bmatrix}. \]
    while the matrix entries inside of each block \((n,m)\) are indexed by \((k,\ell)\). The matrix being block Toeplitz now follows from the earlier observation that \[
A_{(n+1,k),(m+1,\ell)} = A_{(n,k),(m,\ell)},\]
for all \(n,m \in \mathbb Z\), \(k,\ell \in \mathbb Z_{\geq 1}\).  Indeed, this shows that the blocks indexed by \((n,m)\) and \((n+1,m+1)\) are always equal; that is, blocks down the same diagonal are the same.
\end{proof}

\subsection{The special case \(K = 1\)}\label{subsub:K1}
As the matrix \(A(N,K)\) is block Toeplitz with blocks of size \(K \times K\), it makes sense to start by studying the case \(K = 1\). Then, the matrix \(A(N,1)\) is just a real symmetric \((N+1)\times(N+1)\) Toeplitz matrix. That is, \(A\) is of the following form:
\[
  A =  \begin{bmatrix}
a_0           & a_{1} & a_{2} & \cdots  & a_{N} \\
a_1           & a_0              & a_{1} & \cdots  & a_{N-1} \\
a_2           & a_1              & a_0              & \cdots  & a_{N-2} \\
\vdots        & \vdots           & \vdots           & \ddots & \vdots            \\
a_{N}         & a_{N-1}          & a_{N-2}          & \cdots & a_0
\end{bmatrix},
    \] where the sequence \(a_0, a_1, a_2,\ldots\) is given by (see the proof of Lemma \ref{prop:A}) \[a_n = A_{{(0,1)},{(n,1)}} = - 2^{n/2} \int_{-2^{-n}}^{-2^{-n-1}} J(y-1)  \, \mathrm{d}y + 2^{n/2} \int_{-2^{-n-1}}^0 J(y-1)  \, \mathrm{d}y.\] The following closed form expression may be computed:
   \begin{align*}
       a_n = 2^{-n/2} \Big( &2^n \ln (2)-2 \left(2^{n-1}+1\right) \ln \left(2^{n-1}+1\right) \\ &+3 \left(2^n+1\right) \ln \left(2^n+1\right)-\left(2^{n+1}+1\right) \ln \left(2^{n+1}+1\right)\Big).
   \end{align*}
   Note that \(a_n>0\) (by Lemma \ref{prop:A}), and that the series \(\sum_{n=0}^\infty a_n\) converges (by the ratio test, since \(\lim_{n\to\infty} \left|\frac{a_{n+1}}{a_n}\right| = \frac 1 {\sqrt 2}\)).
 See also Fig.\ \ref{fig:anK=1}.

\begin{figure}
    \centering
    \includegraphics[width=0.65\linewidth]{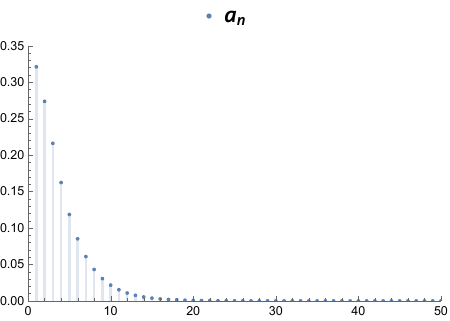}
    \caption{The terms \(a_n\) vanish relatively quickly.}
    \label{fig:anK=1}
\end{figure}

The following lemma, taken from \cite[Lemma 4.1]{Gray2006}, concerns the extremal eigenvalues of real symmetric Toeplitz matrices.\begin{lemma}\label{prop:szegobounds}
Let \((a_n)_{n \geq 0}\) be an absolutely summable sequence of real numbers and define an associated sequence of real symmetric Toeplitz matrices
 \[
A_n = \begin{bmatrix}
a_0           & a_{1} & a_{2} & \cdots  & a_{n}  \\
a_1           & a_0              & a_{1}  & \cdots  & a_{n-1}  \\
a_2           & a_1              & a_0              & \cdots  & a_{n-2}  \\
\vdots        & \vdots           & \vdots           & \ddots & \vdots            \\
a_{n}         & a_{n-1}          & a_{n-2}          & \cdots & a_0
\end{bmatrix}. \]  Putting \(a_{-n}:=a_n\), the associated Fourier series  \[
f(t) = \sum_{n=-\infty}^\infty a_n e^{int} = a_0 +  2\sum_{n=1}^{\infty} a_n \cos(nt), \qquad t \in [0,2\pi]
\]
is real valued, and  \[\min f \leq \lambda \leq \mathrm{max} f\] for each \(n \in \mathbb Z_{\geq 0}\) and each \(\lambda \in \operatorname{Spec}(A_n)\).
\end{lemma}
This result yields an upper bound for the maximal eigenvalue of the matrix \(A\): as the associated Fourier series is given by
\[f(t) = a_0 +  2 \sum_{n=1}^{\infty} a_n \cos(nt), \qquad t \in [0,2\pi],\]
and because all \(a_n > 0\), the series attains its maximum when all the cosines are 1. In fact, it is easy to see that their proof may be slightly modified to obtain the following (stronger) result: \begin{equation}
\lambda_{\mathrm{max}}\big(A(N,1)\big) \leq a_0 +  2 \sum_{n=1}^{N} a_n, \label{eq:upperbound}
\end{equation} by only considering the truncated Fourier series.

A classical lower bound for the maximal eigenvalue of the matrix \(A\) is given by \[
\frac S {N+1} \leq \lambda_{\mathrm{max}}\big(A(N,1)\big),
\] where \(S\) is the sum of all entries in the matrix \(A(N,1)\) (see e.g.\ \cite{Benasseni2009}).   In our case, \begin{equation}
\frac S {N+1}  = a_0 + 2\Big(\frac{N}{N+1} a_1 + \frac{N-1}{N+1} a_2 +  \frac{N-2}{N+1} a_3 + \ldots + \frac{1}{N+1} a_N \Big). \label{eq:lowerbound}
\end{equation} Using both bounds \eqref{eq:upperbound} and \eqref{eq:lowerbound}, combined with the fact that the series \(\sum_{n=1}^{\infty} a_n\) converges, we have proved the following result:
\begin{proposition}[Asymptotic maximal eigenvalue for \(K = 1\)]\label{lemma:eigenvalueSeries}
    Given \(N \geq 1\), we have that the maximal eigenvalue of \(A(N,1)\) may be approximated by
  \[ a_0 + 2\Big(\frac{N}{N+1} a_1 + \frac{N-1}{N+1} a_2 +  \frac{N-2}{N+1} a_3 + \ldots + \frac{1}{N+1} a_N \Big) \leq \lambda_{\mathrm{max}}\big(A(N,1)\big) \leq a_0 +  2 \sum_{n=1}^{N} a_n.\]
  In particular, \[\lim_{N \to \infty } \lambda_{\mathrm{max}}\big(A(N,1)\big) = a_0 +  2 \sum_{n=1}^{\infty} a_n.\]
\end{proposition}
To prove the last statement, we note that one can rewrite the l.h.s.~of the foregoing inequality as $a_0 +  \frac{2}{N+1}\sum_{k=1}^N \sum_{n=1}^k  a_n$.  Here the sums $ \frac{1}{N+1} \sum_{k=1}^N \sum_{n=1}^k  a_n$ equal  $\frac{N+1}{N}$ times the Cesàro means of the partial sums $\sum_{n=1}^N a_n$, henceforth converge to the same limit $\sum_{n=1}^\infty a_n$.

Let us simplify this series. We start by investigating the following sequence of partial sums:
\begin{lemma}\label{lemma:seriesIntegrals}
    Given \(N \geq 1\), we have that \begin{align*}
    \sum_{n=1}^N 2^{-n/2} \left(1+2^n \right) & \ln(1 + 2^n) =
   \frac{ 2^{-N/2} (\sqrt 2  - 1)}{3-2\sqrt 2}\left( 2^{N/2} - 1 \right)  + \sum_{n=1}^N \iota_n
    \\
    & + \frac{\ln 2}{3-2\sqrt 2} \Bigg( 2\sqrt 2 + 2^{N/2} \sqrt 2 \left( (\sqrt 2 - 1)N -1 \right) - 2^{-N/2} \left( \sqrt 2 + (\sqrt 2 - 1) N\right) \Bigg).
    \end{align*}
     The terms \(\iota_n\) are integrals given  by
  \[\iota_n = \int_0^1   \frac{2^{-n/2}(1-x)}{x + 2^{n}} \, \mathrm{d}x > 0.\]
\end{lemma}
\begin{proof}
    The central insight is to rewrite the \(\ln(1+2^n)\) factor in the sum by integrating \(1/(x+2^n)\):
    \begin{align*}
    2^{-n/2} \left(1+2^n \right) \ln(1+2^n) &= 2^{-n/2} \left(1+2^n \right) \left( \ln (2^n) + \int_0^1 \frac{\mathrm{d}x}{x+2^n}\right) \\
    &=  \ln (2) \left( n 2^{-n/2} + n 2^{n/2} \right)   + \int_0^1 \frac{2^{-n/2}(1+2^n)}{x+2^n} \, \mathrm{d}x \\
    &= \ln (2) n \left( 2^{-n/2} + 2^{n/2} \right)   + 2^{-n/2} + \int_0^1 \frac{2^{-n/2}(1-x)}{x+2^n} \, \mathrm{d}x.
    \end{align*}
    From here on, it can be verified that  \begin{align*}
    \sum_{n=1}^N  & n \left( 2^{-n/2} + 2^{n/2} \right)  = \\\frac{1}{3-2\sqrt 2} \Bigg( 2\sqrt 2 + 2^{N/2} \sqrt 2 &\left( (\sqrt 2 - 1)N -1 \right) - 2^{-N/2}  \left( \sqrt 2 + (\sqrt 2 - 1) N\right) \Bigg) \end{align*} and
    \[ \sum_{n=1}^N 2^{-n/2} = \frac{ 2^{-N/2} (\sqrt 2  - 1)}{3-2\sqrt 2}\left( 2^{N/2} - 1 \right).\]
\end{proof}
Using this lemma, we can rewrite the partial sum \(\sum_{n=1}^N a_n\) into a closed form expression plus (a multiple of) the sum of integrals \(\sum_{n=1}^N \iota_n\).
\begin{lemma} \label{lemma:PartialSumK=1}
    Given \(N \geq 1\), we can write the partial sum \(\sum_{n=1}^N a_n \) as \[\sum_{n=1}^N a_n = \alpha + \beta_N + \gamma_N  + (3-2\sqrt 2)  \, \sum_{n=1}^N\iota_n,\] where we have grouped constant terms
    \[\alpha = \sqrt 2 - 1 - (2 + \sqrt 2) \ln 2 + 3 \ln 3 \approx 1.3435,\]
   and terms that exponentially decay to zero as \(N \to \infty\):
    \begin{align*}
    \beta_N = 2^{-N/2} \bigg(
1 - \sqrt 2 - (\sqrt 2 + 1) \ln 2
    +\sqrt 2 \ln \big(1+2^{-N} \big) - \ln \big(1+ 2^{-N-1} \big) < 0
    \bigg)
    \end{align*}
  and \[
  \gamma_N = 2^{N/2} \sqrt{2} \bigg(
  \ln \big( 1 + 2^{-N} \big) - \sqrt{2} \ln \big( 1+ 2^{-N-1} \big)
  \bigg) > 0.
  \]  The terms \(\iota_n\) are the integrals from Lemma \ref{lemma:seriesIntegrals}.
\end{lemma}
\begin{proof} We start with
    \begin{align*}
     \sum_{n=1}^N a_n  = \sum_{n=1}^N 2^{-n/2} \Big( &2^{n} \ln (2)-2 \left(2^{n-1}+1\right) \ln \left(2^{n-1}+1\right) \\ &+3 \left(2^n+1\right) \ln \left(2^n+1\right)-\left(2^{n+1}+1\right) \ln \left(2^{n+1}+1\right)\Big).
   \end{align*}
From this sum, we extract the following partial sum of a geometric series:
   \begin{equation}
       \sum_{n=1}^N 2^{-n/2} \cdot 2^{n} \ln (2) = \frac{\sqrt 2}{\sqrt 2 - 1} \left(2^{N/2}-1 \right) \ln( 2) .\label{eq:K1FirstPart}
   \end{equation}
The remaining terms are rewritten by applying a reindexing step. During this step, we have to account for the beginning and end terms that get added or cut off during this reindexing step, as these are not negligible in the limit \(N\to\infty\). We find that \begin{align}
    \sum_{n=1}^N 2^{-n/2} 2 \big( 2^{n-1} &+ 1 \big) \ln (2^{n-1}  +1)  = 2 \sqrt{2} \ln 2 - 2^{-\frac{N+1}2} 2 \left( 2^{N} + 1 \right) \ln (2^{N}  +1) \notag \\ &+ \sum_{n=1}^N 2^{-\frac{n+1}2} 2  \left( 2^{n} + 1 \right) \ln (2^{n}  +1) \label{eq:reindex1}
\end{align}
and similarly \begin{align}
    \sum_{n=1}^N 2^{-n/2}  \big( 2^{n+1} &+ 1 \big) \ln (2^{n+1}  +1)  = -3  \ln 3 + 2^{-N/2}  \left( 2^{N+1} +1 \right) \ln (2^{N+1}  +1) \notag \\ &+ \sum_{n=1}^N 2^{-\frac{n-1}2}   \left( 2^{n} +1 \right) \ln (2^{n}  +1). \label{eq:reindex2}
\end{align}
The remaining terms are
\begin{equation}
\sum_{n=1}^N 2^{-n/2} 3 (2^n + 1) \ln (2^n +1),
    \label{eq:reindex3}
\end{equation}
and these do not have to be reindexed. We currently have
\[
\sum_{n=1}^N a_n = \eqref{eq:K1FirstPart} - \eqref{eq:reindex1} - \eqref{eq:reindex2} + \eqref{eq:reindex3}.
\]
Let us focus on \(- \eqref{eq:reindex1} - \eqref{eq:reindex2} + \eqref{eq:reindex3}.\) These terms may be rewritten as \[
- \eqref{eq:reindex1} - \eqref{eq:reindex2} + \eqref{eq:reindex3} = B(N) +
(3-2\sqrt 2) \sum_{n=1}^N 2^{-n/2} \left(1+2^n \right) \ln(1 + 2^n),\] where we have collected the beginning and end terms from earlier as \(B(N)\):
\begin{align*}
B(N) &= 3 \ln (3)-2 \sqrt{2} \ln (2) \\
&{}+ 2^{-N/2} \bigg(\sqrt{2} \left(2^N+1\right) \ln \left(2^N+1\right)-\left(2^{N+1}+1\right) \ln \left(2^{N+1}+1\right)\bigg).
\end{align*}
The remaining terms \(\sum_{n=1}^N 2^{-n/2} \left(1+2^n \right) \ln(1 + 2^n)\) may be informally referred to as `the hard part' and we have derived an expression for these in Lemma \ref{lemma:seriesIntegrals}.

Finally, after some rewriting, the partial sum \(\sum_{n=1}^N a_n\) can indeed be found to equal the expression in the statement of this Lemma \ref{lemma:PartialSumK=1}.
\medskip

\noindent \textit{\(\beta_N\) and \(\gamma_N\) exponentially decay to zero as \(N\to\infty\):} That \(\beta_N < 0\) exponentially decays to zero is clear. To show the same is true for \(\gamma_N > 0\), we note that
\[
\ln\big(1+2^{-N}\big) - \sqrt{2} \ln\big(1+2^{-N-1}\big)= \left(1-\frac1{\sqrt2} \right) 2^{-N} + O\big( 2^{-2N} \big) \qquad \text{as \(N \to \infty\).}
\]
\end{proof}
Combining Lemmas \ref{lemma:eigenvalueSeries}, \ref{lemma:seriesIntegrals} and \ref{lemma:PartialSumK=1}, we deduce the following:
\begin{proposition}\label{prop:K=1} The asymptotic maximal eigenvalue approaches
    \[
\lim_{N\to\infty}    \lambda_{\mathrm{max}}\big(A(N,1)\big) = \ln\left( \frac{1024}{729} \right) + 2\alpha + 2 (3-2\sqrt2) \sum_{n=1}^\infty \iota_n \approx 3.1105202.
    \]
\end{proposition}
Below, we explain why the series of integrals \(\sum_{n=1}^\infty \iota_n \) converges and how to obtain approximations for its partial sums.
\begin{figure}
    \centering
 \begin{subfigure}[b]{0.47\textwidth}
    \centering
\includegraphics[width=1\textwidth]{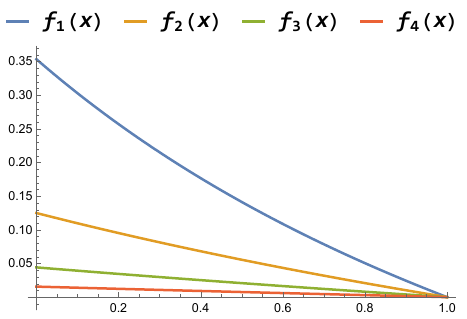}
\end{subfigure}
     \hfill
\begin{subfigure}[b]{0.45\textwidth}
    \centering
\includegraphics[width=1\textwidth]{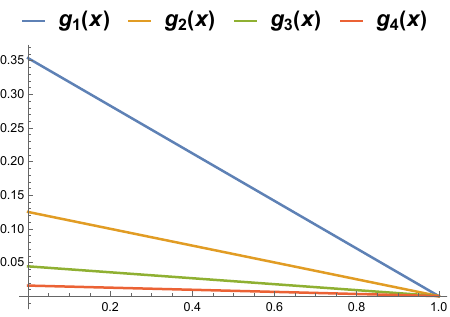}
\end{subfigure}
\caption{The first four functions \(f_n(x) = \dfrac{2^{-n/2}(1-x)}{x+2^n}\) together with their linear upper bounds \(g_n(x) = 2^{-3/2} (1-x)\).}
\label{fig:fnK=1}
\end{figure}

\begin{remark}
    The result of Proposition \ref{prop:K=1} is already \(3.11052 / \pi = 99.01 \%\)
of the way to \(\pi\); but as could be expected, we do not fully approach it.

Nevertheless, it is remarkable how closely we may approximate \(\pi\) in the special case \(K=1\), suggesting that the parameter \(K\), governing the number of horizontal translates of the Haar wavelets, plays a less significant role than \(N\), controlling how fine the wavelets are.
\end{remark}

\paragraph{The series \(\sum_{n=1}^\infty \iota_n\).}
 That we indeed cannot ``fully get to \(\pi\)'' may be clarified by considering a specific upper bound for the series of integrals \(\sum_{n=1}^\infty \iota_n\). Recall that \(\iota_n = \int_0^1 f_n,\) where \[f_n(x) =   \frac{2^{-n/2}(1-x)}{x+2^n}, \qquad x\in [0,1].\] These functions are all positive, decreasing, convex,  and their graphs connect the points \((0,2^{-3n/2})\) and \((1,0)\), as may be observed in Fig.\ \ref{fig:fnK=1}. As such, we can dominate them by the following sequence of linear functions: \[g_n(x) =   {2^{-3n/2}(1-x)}, \qquad x\in [0,1],\] with corresponding integrals \[\kappa_n = \int_0^1 g_n = 2^{-1-\frac{3n}{2}}.\]

It follows that the asymptotic maximal eigenvalue has the following \textit{exact} upper bound: \begin{align*}
    M_0 &= \ln\left( \frac{1024}{729} \right) + 2\alpha + 2 (3-2\sqrt2) \sum_{n=1}^\infty \kappa_n \\
    &= 2 (3 - \sqrt 2) \ln 2 + \frac{18 \sqrt 2 - 19}{7} \\
    &\approx 3.12063,
\end{align*} demonstrating that we indeed do not get to \(\pi\).

Sharper upper bounds may be found by computing the first \(N\) integrals \(\iota_n\) exactly and bounding the rest (\(n \geq N+1\)) by \(\kappa_n\):
\[
M_N =\ln\left( \frac{1024}{729} \right) + 2\alpha + 2 (3-2\sqrt2) \Big( \sum_{n=1}^N \iota_n + \sum_{n=N+1}^\infty \kappa_n \Big).
\] For example:
\begin{align*}
M_1 &= \frac{64 - 33 \sqrt 2}{28} + (18 - 11 \sqrt 2) \ln 2 +   3 (-4 + 3 \sqrt 2) \ln 3 \approx 3.11248, \\
M_2 &= \frac{60 \sqrt 2 - 61}{56} + 5 (2 \sqrt 2 - 3) \ln 5 +
 3 (3 \sqrt 2 - 4) \ln 6 \approx 3.11088,
\end{align*} and so on.

\subsection{Outlook: The general case \texorpdfstring{\(K \geq 1\)}{K ≥ 1}}
For general \(K\), the matrix \(A\) is no longer necessarily Toeplitz, but block Toeplitz. Luckily, for block Toeplitz matrices we have the following generalization of Lemma \ref{prop:szegobounds}: \begin{thm}[{See \cite[Corollary 3.5]{Miranda2000} or \cite[Theorem 3]{Gazzah2001}}] \label{prop:blockszego}
     Let \((A_n)_{n \geq 0}\) be a sequence of matrices in \(\mathbb C^{K \times K}\) such that for all \(k,\ell \in \{1,\ldots, K\}\) we have that the sequence \(((A_n)_{k,\ell})_{n\geq 0}\) is absolutely summable in \(\mathbb C\) and consider the sequence of Hermitian block Toeplitz matrices
\[
\mathcal A_n = \begin{bmatrix}
A_0           & A_{1}^\ast & A_{2}^\ast & \cdots  & A_{n}^\ast \\
A_1           & A_0              & A_{1}^\ast & \cdots  & A_{n-1}^\ast \\
A_2           & A_1              & A_0              & \cdots  & A_{n-2}^\ast \\
\vdots        & \vdots           & \vdots           & \ddots & \vdots            \\
A_{n}         & A_{n-1}          & A_{n-2}          & \cdots & A_0
\end{bmatrix} \, \in \, \mathbb C^{(n+1)K \times (n+1)K}.
\]   Putting \(A_{-n} = A_n^\ast\), the \(K \times K\) matrix valued Fourier series \[
F(t) = \sum_{n=-\infty}^{\infty} A_n e^{int}, \qquad t \in [0,2\pi]
\] is Hermitian and we have
     \begin{align*}
    \lim_{n \to \infty} \lambda_{\max} (\mathcal A_n) &= \max_{t \in [0,2\pi]}  \lambda_{\max} (F(t)) \\
    \lim_{n \to \infty} \lambda_{\min} (\mathcal A_n) &= \min_{t \in [0,2\pi]} \lambda_{\min}  (F(t)).
\end{align*}
\end{thm}
Similarly as in Lemma \ref{prop:szegobounds}, we also have that these asymptotics for the extreme eigenvalues are effective bounds on all eigenvalues: for each \(n \in \mathbb Z_{\geq 0}\) and each \(\lambda \in \operatorname{Spec}(\mathcal A_n)\), we have \[
\min_{t \in [0,2\pi]} \lambda_{\min}  (F(t)) \leq \lambda \leq  \max_{t \in [0,2\pi]} \lambda_{\max} (F(t)).
\]
Let us define \(F_K(t)\) as the \(K\times K\) matrix valued Fourier series \[
F_K(t)= \sum_{n=-\infty}^{\infty} A_n(K) e^{int}, \qquad t \in [0,2\pi].
\] According to Theorem \ref{prop:blockszego}, we have the result
\[\lim_{N \to \infty} \lambda_{\max} \big(A(N,K) \big) = \max_{t \in [0,2\pi]}  \lambda_{\max} \big( F_K(t) \big).\] The situation here is considerably more complex than the special case \(K=1\), as we now need to compute the maximal eigenvalue of a matrix valued Fourier series (which is trivial only in the case of a \(1\times 1\) matrix). Nonetheless, we conjecture that for any \(K \geq 1\), similar to the case \(K=1\), the maximum is  achieved at \(t=0\): \[
\max_{t \in [0,2\pi]}  \lambda_{\max} \big( F_K(t) \big) = \lambda_{\max} \big( F_K(0) \big).
\]  Fig.\ \ref{fig:fourierMaxK} shows the  function \(\lambda_{\max} \big( F_K(t) \big)\) for the values \(K = 2\) and \(K=4\), providing numerical evidence of the conjecture that the maximum is indeed achieved at \(t=0.\) The graphs for higher values of \(K\) are visually indistinguishable, though the maximum gets closer and closer to \(\pi\).
\begin{figure}
    \centering
 \begin{subfigure}[b]{0.48\textwidth}
    \centering
\includegraphics[width=1\textwidth]{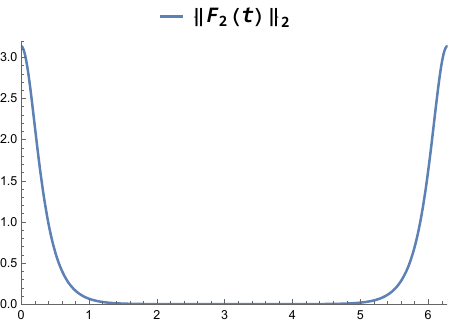}
\end{subfigure}
     \hfill
\begin{subfigure}[b]{0.48\textwidth}
    \centering
\includegraphics[width=1\textwidth]{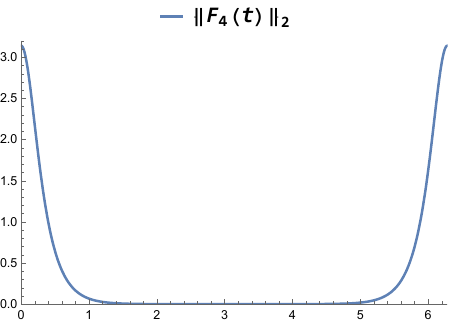}
\end{subfigure}
\caption{The graphs of \(\lambda_{\max}\big( F_2(t)\big)\) and \(\lambda_{\max}\big( F_4(t)\big)\). The graphs corresponding to different values of \(K\) are very similar.}
\label{fig:fourierMaxK}
\end{figure}

\begin{table}
    \centering
    \begin{subtable}{0.32 \textwidth}
        \begin{tabular}{c | c}
        \(K\) & \(\lambda_{\max}\big(F_K(0)\big)\) \\ \hline
   1 &  3.1105201 \\  2 &  3.1330806 \\  3 &  3.1377294 \\  4 &  3.1394026 \\  5 &  3.1401857 \\  6 &  3.1406136 \\  7 &  3.1408725 \\  8 &  3.1410408 \\  9 &  3.1411564 \\  10 &  3.1412391 \\
   11 &  3.1413004 \\  12 &  3.1413470 \\  13 &  3.1413833 \\  14 &  3.1414121 \\  15 &  3.1414354 \\  16 &  3.1414544 \\  17 &  3.1414702 \\  18 &  3.1414834 \\  19 &  3.1414946 \\  20 &  3.1415042
        \end{tabular}
    \end{subtable}
    \begin{subtable}{0.32 \textwidth}
             \begin{tabular}{c | c}
        \(K\) & \(\lambda_{\max}\big(F_K(0)\big)\) \\ \hline
 21 &  3.1415124 \\  22 &  3.1415195 \\  23 &  3.1415258 \\  24 &  3.1415312 \\  25 &  3.1415360 \\  26 &  3.1415403 \\  27 &  3.1415441 \\  28 &  3.1415475 \\  29 &  3.1415506 \\  30 &  3.1415534 \\  31 &  3.1415559 \\  32 &  3.1415581 \\  33 &  3.1415602 \\  34 &  3.1415621 \\  35 &  3.1415638 \\  36 &  3.1415654 \\  37 &  3.1415669 \\  38 &  3.1415682 \\  39 &  3.1415694 \\  40 &  3.1415706
        \end{tabular}
    \end{subtable}
    \begin{subtable}{0.32 \textwidth}
               \begin{tabular}{c | c}
        \(K\) & \(\lambda_{\max}\big(F_K(0)\big)\) \\ \hline
  41 &  3.1415717 \\  42 &  3.1415727 \\  43 &  3.1415736 \\  44 &  3.1415745 \\  45 &  3.1415753 \\  46 &  3.141576 \\  47 &  3.1415768 \\  48 &  3.1415774 \\  49 &  3.1415780 \\  50 &  3.1415786 \\  51 &  3.1415792 \\  52 &  3.1415797 \\  53 &  3.1415802 \\  54 &  3.1415807 \\  55 &  3.1415811 \\  56 &  3.1415815 \\  57 &  3.1415819 \\  58 &  3.1415823 \\  59 &  3.1415826 \\  60 &  3.1415830
        \end{tabular}
    \end{subtable}
    \caption{Numerical evidence for \(\lambda_{\max}\big(F_K(0)\big) \longrightarrow \pi\) as \(K\to\infty\).}
\label{tab:fourierK}
\end{table}

\begin{remark} In addition, we remark that it is not necessary to prove that the maximum is achieved at \(t=0\). It suffices to show \(
    \lambda_{\max}\big(F_K(0)\big) \longrightarrow \pi,
    \) as this would directly imply
    \[
   \max_{t\in [0,2\pi]} \lambda_{\max}\big(F_K(t)\big) \longrightarrow \pi.
    \]
     To establish this, it again suffices to prove that for any \(\delta > 0\), there exists a sufficiently large \(K \geq 1\) such that \(\pi - \delta < \lambda_{\max}\big(F_K(0)\big).\) (The reason why is similar to the argument presented in Remark \ref{rem:underlyingPhysics}.)
\end{remark}

  For reference, the \(K \times K\) matrix valued series \(F_K(0)\) is \[
F_K(0) = \sum_{n=-\infty}^{\infty} A_n(K) = A_0(K) + \sum_{n=1}^{\infty} \big(A_n(K) + A_n^\mathsf T(K) \big).
\]
Numerical experiments clearly suggest that indeed
\[
\lambda_{\max}\big(F_K(0)\big) \longrightarrow \pi \qquad \text{as \(K\to \infty\),}
\] see Table \ref{tab:fourierK}.  As a self-consistency check, we note the correspondence between the exact asymptotic eigenvalue in the special case \(K=1\) obtained in Proposition \ref{prop:K=1}, and the one estimated in Table \ref{tab:fourierK}.
\begin{remark}
      For the numerical experiments (Fig.\ \ref{fig:fourierMaxK} and Table \ref{tab:fourierK}) we approximate \(F_K(t)\) by the partial sum \[\sum_{n=-50}^{50}A_n(K)e^{int}.\] In particular, \(F_K(0)\) is approximated by \[A_0(K) + \sum_{n=1}^{50} \big(A_n(K) + A_n^\mathsf T(K) \big).\]
\end{remark}
\begin{outlook}
    Although a formal proof for \(\lambda_{\max}\big(F_K(0)\big) \longrightarrow \pi\) as \(K \to \infty\) currently eludes us, we expect that it may be done with the appropriate mathematical tools.
\end{outlook}


\section{Bumpification (Step 2)}\label{sec:bump}
Assume now that for a given \(\eta \in (\sqrt2 -1, 1)\), we can find a resolution \(\{N_0,N_1,K\}\) and initial test functions, in terms of Haar wavelets \eqref{eq:expansion}, which present an exact solution to the system of equations \eqref{eq:28}. That is, in this section we \textit{assume} Conjecture \ref{conj:soluble} is true (possibly by assuming Conjecture \ref{conj:max} is true), and demonstrate that these initial test functions may be transformed into \(C^{\infty}\) versions consistent with the smoothness requirements of QFT.

The process is as follows: In Subsection \ref{sub:bumpHaar}, we introduce \(C^\infty\) versions of the Haar wavelets  by smoothing out their points of discontinuity, using a small tuning parameter \(\varepsilon>0\). Then, in Subsection \ref{sub:bumpTest}, we demonstrate that by choosing \(\varepsilon\) sufficiently small, these ``bumpified'' Haar wavelets closely approximate the original wavelets (in the appropriate sense). This ensures that, after substituting the \(C^\infty\) wavelets for the original ones in the expansion \eqref{eq:expansion}, the system of equations \eqref{eq:28} remains satisfied up to the desired degree of precision, controlled by \(\varepsilon.\)

\subsection{Bumpified Haar wavelets}\label{sub:bumpHaar}
To construct smooth ``bumpified'' versions of the Haar wavelets (still compactly supported), following \cite{Mc2010}, the authors of \cite{Dudal2023} employ the \textit{Planck-taper window function}.
\begin{definition}
    The \textit{basic Planck-taper window function} with support on the interval
\([0,1]\)
is defined by
\[
s^\varepsilon (x) =
\begin{cases}
    \left[1+\exp\left(\dfrac{\varepsilon(2x-\varepsilon)}{x(x-\varepsilon)}\right) \right]^{-1} & \text{ if } 0 < x < \varepsilon \\
    1 & \text{ if } \varepsilon \leq x \leq 1 - \varepsilon \\
    \left[1+\exp\left(\dfrac{\varepsilon(-2x-\varepsilon+2)}{(x-1)(x+\varepsilon-
    1)}\right) \right]^{-1}
    & \text{ if } 1-\varepsilon < x < 1 \\
    0 & \text{ otherwise.}
\end{cases}
\]
Here, \(\varepsilon > 0\) is a sufficiently small tuning parameter.
\end{definition}
This function is the \(C^\infty\) version of the `basic rectangle',
\[
r(x) =
\begin{cases}
    1 & \text{ if } 0 \leq x < 1 \\
    0 & \text{ otherwise,}
\end{cases}
\]
smoothing over the  point of discontinuity \(x = 0\)  by making use of the transition function \(x \mapsto \left[1+\exp\left(\frac{\varepsilon(2x-\varepsilon)}{x(x-\varepsilon)}\right) \right]^{-1}\), and symmetrically for the other point \(x=1\). See Fig.\ \ref{fig:planck}.

\begin{figure}
    \centering
\begin{tikzpicture}[scale=3, line width = .8pt]
		\def\xmin{-0.5}
		\def\xmax{1.5}
		\def\ymin{-0.25}
		\def\ymax{1.35}
        \def\eps{0.15}

        \draw[dashed] (0,1) -- (1,1) -- (1,0);
        \draw[stealth-stealth] (0,1.075) --node[midway,above]{\(\varepsilon\)} (\eps,1.075);
        \draw[stealth-stealth] (1-\eps,1.075) --node[midway,above]{\(\varepsilon\)}  (1,1.075);

        \draw[very thick, red] plot [domain=\xmin:0.015] (\x, 0.005);
		\draw [very thick,red] plot [domain=0.015:0.1495,samples=50] (\x,{0.005 + 0.99/(1+exp(
  (\eps*(2*\x-\eps))/(\x*(\x - \eps))
  ))});
        \draw[very thick, red] plot [domain=0.1485:1-0.1485] (\x, 1-0.005);
        \draw [very thick,red] plot [domain=1-0.149:1-0.0175,samples=50] (\x,{0.005 + 0.99/(1+exp(
  (\eps*(2*(1-\x)-\eps))/((1-\x)*((1-\x) - \eps))
  ))});
        \draw[very thick, red] plot [domain=1-0.0175:\xmax-0.035] (\x, 0.005);

        \draw[very thick] (1,-0.02) --node[midway,below]{\(1\)} (1,0.02);

        \node[below left, inner sep = 2 pt] at (0,0) {\(0\)};
        \draw[very thick] (-0.02,1) --node[midway,left]{\(1\)} (0.02,1);

        \draw[-latex]  (\xmin,0) -- (\xmax,0) node[below]{\(x\)};
		\draw[-latex]  (0,\ymin) -- (0,\ymax) ;

\end{tikzpicture}
\caption{The basic Planck-taper window function \(s^\varepsilon\).}  \label{fig:planck}
\end{figure}
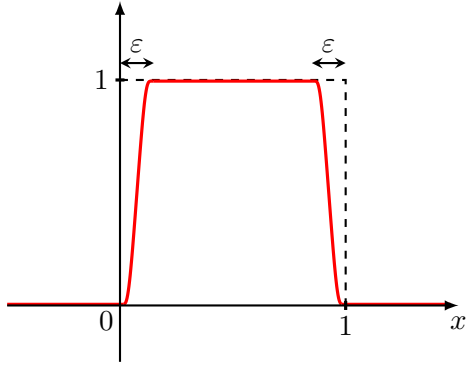

 The deviation between the window function \(s^\varepsilon\) and the basic rectangle \(r\) may be made arbitrarily small by tuning \(\varepsilon\). In fact, if we measure this deviation using any \(L^p\) norm (except \(L^\infty\)), it is proportional to the tuning parameter \(\varepsilon\).

\begin{proposition}\label{thm:Lp}
  Let \(1 \leq p < \infty\), \(0 < \varepsilon < 1/2\), and write \(\lVert \cdot \rVert_{p}\) for the \(L^p\) norm. Then,\[
    \lVert r - s^\varepsilon \rVert^p_{p} = \alpha_p \,  \varepsilon \leq \varepsilon,
    \]
    where \[\alpha_p := \displaystyle{ 2 \int\limits_0^1  \bigg( 1 -  \frac1{1+\exp\left(\frac{2y-1}{y(y-1)}\right)}\bigg)^p \, \mathrm{d}y \,\,\, \in \,\,\,(0,1\mathclose].
    }\]
\end{proposition}
\begin{proof}
Since \(r \equiv s^\varepsilon\) outside \((0,\varepsilon) \cup (1-\varepsilon, 1)\), we deduce
    \begin{align*} \left\lVert r - s^\varepsilon \right\rVert^p_{p}
    &= \int\limits_0^\varepsilon \left\lvert r(x) - s^\varepsilon(x) \right\rvert^p \, \mathrm{d}x + \int\limits_{1-\varepsilon}^1 \left\lvert r(x) - s^\varepsilon(x) \right\rvert^p \, \mathrm{d}x \\
       &= 2 \int\limits_0^\varepsilon \lvert r(x) - s^\varepsilon(x) \rvert^p \, \mathrm{d}x.
    \end{align*}
    Now, substituting \(y = x/\varepsilon\)  yields
    \[
   \lVert r - s^\varepsilon \rVert^p_{p} =  2 \int\limits_0^1 \bigg( 1 -  \frac1{1+\exp\left(\frac{2y-1}{y(y-1)}\right)}\bigg)^p \, \varepsilon \, \mathrm{d}y  = \alpha_p \varepsilon.
    \]

    To prove that this constant satisfies \(0 < \alpha_p \leq 1,\)  simply note that \(\alpha_p\) is decreasing in \(p\), and so it suffices to show that \(\alpha_1 = 1\). To prove the identity
    \[
     \int\limits_0^1 2\bigg( 1 -  \frac1{1+\exp\left(\frac{2y-1}{y(y-1)}\right)}\bigg) \, \mathrm{d}y  = 1,
    \] denote the integrand by \(f(y)\), and use the fact that \(f(y) + f(1-y) = 2\). \end{proof}

\begin{remark}
Proposition \ref{thm:Lp} trivially fails for \(p = \infty\), as a uniformly convergent sequence of continuous functions cannot converge to a discountinuous function (like \(r\)).
\end{remark}

\noindent Using the window function \(s^\varepsilon\), we introduce the \(C^\infty\) version of the mother Haar wavelet \(\psi\), by similarly smoothing over its points of discontinuity. See Fig.\ \ref{fig:bump}.

\begin{figure}
    \centering
\begin{tikzpicture}[scale=3, line width = .8pt]
		\def\xmin{-0.5}
		\def\xmax{1.5}
		\def\ymin{-1.05}
		\def\ymax{1.35}
        \def\eps{0.15}

        \draw[dashed] (0,1) -- (0.5,1) -- (0.5,-1) -- (1,-1) -- (1,0);

        \draw[very thick, red] (\xmin, 0.005) -- (0.0075,0.005);
       	\draw [very thick,red] plot [domain=0.0075:0.0745,samples=50] (\x,{0.005 + 0.99/(1+exp(
  (\eps*(2*2*\x-\eps))/(2*\x*(2*\x - \eps))
  ))});
       \draw[very thick, red] plot [domain=0.074:0.5-0.074] (\x, 0.995);
        \draw [very thick,red] plot [domain=0.5-0.0745:0.5-0.007,samples=50] (\x,{0.005 + 0.99/(1+exp(
  (\eps*(2*2*(0.5-\x)-\eps))/(2*(0.5-\x)*(2*(0.5-\x) - \eps))
  ))});
         \draw[very thick, red] plot [domain=0.5-0.0075:0.5+0.0075] (\x, 0.005);
      	\draw [very thick,red] plot [domain=0.5+0.007:0.5+0.0745,samples=50] (\x,{0.005 - 1/(1+exp(
  (\eps*(2*2*(\x-0.5)-\eps))/(2*(\x-0.5)*(2*(\x-0.5) - \eps))
  ))});
       \draw[very thick, red] plot [domain=0.5+0.074:1-0.074] (\x, -0.995);
    \draw [very thick,red] plot [domain=1-0.0745:1-0.007,samples=50] (\x,{-0.005 - 0.99/(1+exp(
  (\eps*(2*2*(0.5-(\x-0.5))-\eps))/(2*(0.5-(\x-0.5))*(2*(0.5-(\x-0.5)) - \eps))
  ))});
        \draw[very thick, red]   (0.995,0.005) -- (\xmax-0.05,0.005);

        \draw[very thick] (1,-0.02) --node[midway,above]{\(1\)} (1,0.02);

        \node[below left, inner sep = 2 pt] at (0,0) {\(0\)};
        \draw[very thick] (-0.02,1) --node[midway,left]{\(1\)} (0.02,1);
        \draw[very thick] (-0.02,-1) --node[midway,left]{\(-1\)} (0.02,-1);

        \draw[-latex]  (\xmin,0) -- (\xmax,0) node[below]{\(x\)};
		\draw[-latex]  (0,\ymin) -- (0,\ymax) ;

\end{tikzpicture}
\caption{The mother bumpified Haar wavelet \(\sigma^\varepsilon = \sigma^\varepsilon_{0,0}.\)} \label{fig:bump}
\end{figure}
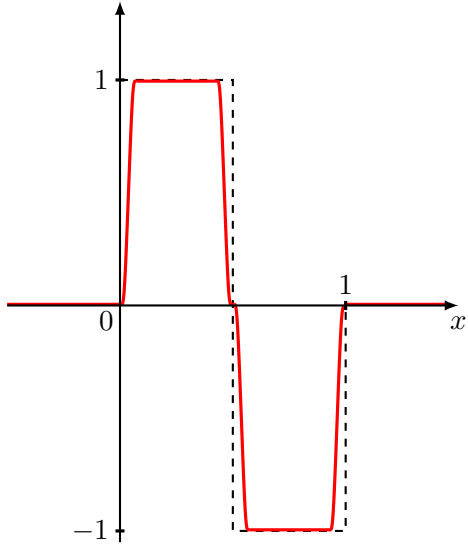

\begin{definition}[Bumpified Haar wavelets]
    The \textit{mother bumpified Haar wavelet} with support on the interval \([0,1]\) is defined by
    \[
    \sigma^\varepsilon(x) = \begin{cases}
        +s^\varepsilon(2x) & \text{ if } 0 \leq x < \dfrac12 \\
        -s^\varepsilon(2x-1) & \text{ if } \dfrac12 \leq x < 1\\
        0 & \text{ otherwise.}
    \end{cases}
    \]
    Bumpified versions  of the Haar wavelets \(\psi_{n,k}\), supported on the intervals \(\overline{I_{n,k}} = [k2^{-n}, (k+1)2^{-n}]\), may then be defined by \[\sigma_{n,k}^\varepsilon(x) = \begin{cases}
        2^{n/2} \sigma^\varepsilon(2^n x - k) & \text{ if } x\in I_{n,k} \\
        0 & \text{ otherwise.}
    \end{cases} \]
\end{definition}

Analogously as in Proposition \ref{thm:Lp}, we can control the deviation between \(\psi_{n,k}\) and \(\sigma_{n,k}\) by tuning \(\varepsilon\), and deduce similar bounds.

\begin{proposition}\label{thm:LpHaar}
    Let \(1\leq p < \infty\). Then, \[
    \lVert \psi - \sigma^\varepsilon \rVert_p^p = \alpha_p \, \varepsilon \leq \varepsilon.
    \]
    Furthermore, for all \(k,n \in \mathbb Z\) we have \[
    \lVert \psi_{n,k} - \sigma^\varepsilon_{n,k}  \rVert_p^p =  2^{n(p/2 -1)} \alpha_p \, \varepsilon .
    \]
\end{proposition}
\begin{proof}
    Using that \[
    \psi(x) = \begin{cases}
         +r(2x) & \text{ if } 0 \leq x < \dfrac12 \\
        -r(2x-1) & \text{ if } \dfrac12 \leq x < 1\\
        0 & \text{ otherwise,}
    \end{cases}
    \]
and by applying appropriate substitutions, we may transform \(\lVert \psi - \sigma^\varepsilon \rVert_p^p\) into \(\lVert r - s^\varepsilon \rVert_p^p\):
\begin{align*}
    \lVert \psi - \sigma^\varepsilon \rVert_p^p
    &= \int\limits_0^{1/2}\big\lvert r(2x) - s^\varepsilon(2x) \big\rvert^p \, \mathrm{d}x + \int\limits_{1/2}^1 \big\lvert r(2x-1) - s^\varepsilon(2x-1) \big\rvert^p \, \mathrm{d}x \\
    &= \frac12 \int\limits_0^{1}\big\lvert r(y) - s^\varepsilon(y) \big\rvert^p \, \mathrm{d}y + \frac12 \int\limits_{0}^1 \big\lvert r(y) - s^\varepsilon(y) \big\rvert^p \, \mathrm{d}y \\
    &= \lVert r - s^\varepsilon \rVert_p^p.
\end{align*}
 Similarly, applying the substitution \(y = 2^n x - k\) gives
\begin{align*}
    \lVert \psi_{n,k} - \sigma_{n,k}^\varepsilon \rVert_p^p
    &=  \int_{I_{n,k}} \big\lvert 2^{n/2} \psi (2^n x - k) - 2^{n/2} \sigma^\varepsilon(2^n x - k) \big\rvert^p \, \mathrm{d}x \\
    &= 2^{ {np}/{2}} \int\limits_0^1 \big\lvert \psi (y) - \sigma^\varepsilon(y) \big\rvert^p \,\, \frac{\mathrm{d}y}{2^n}  = 2^{n(p/2 -1)} \,  \lVert \psi - \sigma^\varepsilon \rVert_p^p = 2^{n(p/2 -1)}  \, \lVert r - s^\varepsilon \rVert_p^p.
\end{align*}
Both results now follow from Proposition \ref{thm:Lp}.
\end{proof}

Lemma \ref{prop:haar} remains valid for these bumpified Haar wavelets, and the proof follows is analogous.
\begin{lemma}\label{prop:bumpHaar}
 For all \(n,k \in \mathbb Z\) we have \[\displaystyle {\sigma^\varepsilon_{n,k}(-x) = -\sigma^\varepsilon_{n,-k-1}(x)}.\]
\end{lemma}

\subsection{The final test functions }\label{sub:bumpTest}
The final step in the procedure is to replace each occurence of the Haar wavelet \(\psi_{n,k}(x)\) by its bumpified counterpart \(\sigma^\varepsilon_{n,k}\) in the expansion \eqref{eq:expansion} of the \textit{initial} set of test functions \((\widetilde f, \widetilde f')\), \((\widetilde g, \widetilde g')\) obtained in Section \ref{sec:probStat}, where we \textit{assume} that these preliminary test functions satisfy the system of equations \eqref{eq:28} \textbf{exactly} for a given \(\eta \in (\sqrt 2 - 1)\) (so assuming Conjecture \ref{conj:soluble} is true).  To certainly satisfy the locality requirement of Alice and Bob's test functions, we implement a small finite translation to the test functions (also using \(\varepsilon\)).  Specifically, we define the new test functions as follows:

\begin{align} \label{eq:Bumptest}
   \begin{array}{c}
\displaystyle{  f_j^{(\varepsilon)}(x)  := \sum_{n = N_0}^{N_1} \sum_{k = -K}^{-1} f_{j}(n,k) \, \sigma^\varepsilon_{n,k}(x+\varepsilon)}, \quad g_j^{(\varepsilon)}(x) := \sum_{m = N_0}^{N_1} \sum_{\ell = 0}^{K-1} g_{j}(m,\ell) \, \sigma^\varepsilon_{n,k}(x-\varepsilon),
   \end{array}
\end{align}
for \(j \in \{1,2\}\), with analogous expressions for the other two test functions \(f'\) and \(g'\).  The key idea is that, by taking \(\varepsilon > 0\) sufficiently small, these new test functions remain close to being normalized and nearly satisfy the necessary inner product relations:
\[
 \Braket{ {f}^{(\varepsilon)} | {g}^{(\varepsilon)} } \approx \Braket{ f'^{(\varepsilon)} | {g}^{(\varepsilon)} }  \approx \Braket{ f^{(\varepsilon)} | g'^{(\varepsilon)} }  \approx - \Braket{ f'^{(\varepsilon)} | g'^{(\varepsilon)} } \approx  -i \frac{\sqrt2 \eta}{1+\eta^2},
\] while still satisfying the smoothness conditions required for them to be valid test functions in the context of QFT. However, these test functions need to be \textit{exactly} normalized, not just approximately. Therefore, we introduce the following normalization factors:
\begin{align*}
    a^{(\varepsilon)} := \Braket{f^{(\varepsilon)}|f^{(\varepsilon)}}^{-1/2}, &\quad a'^{(\varepsilon)} := \Braket{f'^{(\varepsilon)}|f'^{(\varepsilon)}}^{-1/2} \\
    b^{(\varepsilon)} := \Braket{g^{(\varepsilon)}|g^{(\varepsilon)}}^{-1/2}, &\quad b'^{(\varepsilon)} := \Braket{g'^{(\varepsilon)}|g'^{(\varepsilon)}}^{-1/2}.
\end{align*}
If \(\varepsilon\) is small, these factors should all be close to \(1\). We conclude by defining \[ a^{(\varepsilon)}f^{(\varepsilon)}, \quad a'^{(\varepsilon)}f'^{(\varepsilon)}, \quad b^{(\varepsilon)}g^{(\varepsilon)}, \quad b'^{(\varepsilon)}g'^{(\varepsilon)}\] as our  \textbf{final set of test functions} (so now including the small translation) .

To demonstrate that the proposed method of substituting wavelets works, we have to prove the following limit:
\begin{proposition} We have
    \[\displaystyle{\lim_{\varepsilon \to 0} \, \, \Braket{ a^{(\varepsilon)} {f}^{(\varepsilon)} | b^{(\varepsilon)} {g}^{(\varepsilon)} } = -i \frac{\sqrt2 \eta}{1+\eta^2} },\] and similar for the other three inner products.
\end{proposition}
Since \[\Braket{ a^{(\varepsilon)} {f}^{(\varepsilon)} | b^{(\varepsilon)} {g}^{(\varepsilon)} } = a^{(\varepsilon)} b^{(\varepsilon)} \Braket{  {f}^{(\varepsilon)} |  {g}^{(\varepsilon)} },\] and similar for the other three inner products, it suffices to prove the following two statements:
\begin{itemize}
    \item The normalization constants \(a^{(\varepsilon)}\), \(b^{(\varepsilon)}\), \(a'^{(\varepsilon)}\) and \(b'^{(\varepsilon)}\) converge to \(1\) as \(\varepsilon \to 0\).
    \item \(\displaystyle{\lim_{\varepsilon \to 0} \, \, \Braket{  {f}^{(\varepsilon)} |{g}^{(\varepsilon)} } =  -i \frac{\sqrt2 \eta}{1+\eta^2}},\) and similar for the other three inner products.
\end{itemize}  These statements are proved in Lemmas \ref{lemma:normalizationConstants} and \ref{lemma:innerProducts} below.
\begin{lemma}\label{lemma:normalizationConstants}
     We have
    \[\lim_{\varepsilon \to 0} \, \, a^{(\varepsilon)} = 1,\] and similar for the other three normalization factors.
\end{lemma}
\begin{proof}
    It suffices to show that \[\lim_{\varepsilon \to 0} \, \, \Braket{  {f}^{(\varepsilon)} |  f^{(\varepsilon)} } = 1.\] To do this, we make use of the assumption that the preliminary test functions are already normalized, and expand:
    \begin{align*} \left\lvert \Braket{f^{(\varepsilon)} |f^{(\varepsilon)} } -  1 \right\rvert  &=
\left\lvert \Braket{f^{(\varepsilon)} |f^{(\varepsilon)} } -  \Braket{\widetilde f|\widetilde f} \right\rvert \\
&= \left\lvert \int \left(
f_1^{(\varepsilon)\,2} - \widetilde f_1^2 + f_2^{(\varepsilon)\,2} - \widetilde {f}_2^2
\right) \, \mathrm d x \right\rvert \\
&\leq \int \left\lvert f_1^{(\varepsilon)\,2} - \widetilde f_1^2 \right\rvert \, \mathrm d x +
\int \left\lvert f_2^{(\varepsilon)\,2} - \widetilde f_2^2 \right\rvert \, \mathrm d x.
\end{align*}
Here,
\begin{align*}
    \int \left\lvert f_1^{(\varepsilon)\,2} - \widetilde f_1^2 \right\rvert \, \mathrm d x
    &= \int \left\lvert (f_1^{(\varepsilon)} + \widetilde f_1)(f_1^{(\varepsilon)} - \widetilde f_1) \right\rvert \, \mathrm d x \\
    &\leq \max \left\lvert f_1^{(\varepsilon)} + \widetilde{f}_1 \right\rvert \cdot  \int \left\lvert f_1^{(\varepsilon)} - \widetilde f_1 \right\rvert \, \mathrm d x \\
    &= \max \left\lvert f_1^{(\varepsilon)} + \widetilde{f}_1 \right\rvert \cdot \left\lVert  f_1^{(\varepsilon)} - \widetilde f_1 \right\rVert_1,
\end{align*} with \(f_1^{(\varepsilon)} + \widetilde{f}_1\) a sum of \(\varepsilon\)-uniformly bounded functions, and
\[ \left\lVert  f_1^{(\varepsilon)} - \widetilde f_1 \right\rVert_1  \leq \sum_{n=N_0}^{N_1} \sum_{k=-K}^{-1} \left\lvert f_1(n,k) \right\rvert \left\lVert  \sigma_{n,k}^{\varepsilon}({\cdot}+\varepsilon) - \psi_{n,k} \right\rVert_1 \longrightarrow 0 \qquad \text{as \(\varepsilon\to 0\)}.
\]
Indeed,
\begin{align*}
    \left\lVert  \sigma_{n,k}^{\varepsilon}({\cdot}+\varepsilon) - \psi_{n,k} \right\rVert_1 &\leq
    \left\lVert  \sigma_{n,k}^{\varepsilon}({\cdot}+\varepsilon) - \psi_{n,k}(\cdot + \varepsilon) \right\rVert_1 + \left\lVert  \psi_{n,k}(\cdot + \varepsilon) - \psi_{n,k} \right\rVert_1 \\
    &=  \left\lVert  \sigma_{n,k}^{\varepsilon} - \psi_{n,k} \right\rVert_1 + \left\lVert  \psi_{n,k}(\cdot + \varepsilon) - \psi_{n,k} \right\rVert_1
\end{align*}
and both terms converge to zero as \(\varepsilon \to 0\) by Proposition \ref{thm:LpHaar} and the fact that translation is \(L_1\)-continuous, respectively.

Therefore, \[\displaystyle\int \left\lvert f_1^{(\varepsilon)\,2} - \widetilde f_1^2 \right\rvert \, \mathrm d x \longrightarrow 0 \qquad \text{as \(\varepsilon\to 0\).}\] Similarly, \[\displaystyle\int \left\lvert f_2^{(\varepsilon)\,2} - \widetilde f_2^2 \right\rvert \, \mathrm d x \longrightarrow 0 \qquad \text{as \(\varepsilon\to 0\)}.\] This completes the proof.
\end{proof}

\begin{lemma}\label{lemma:innerProducts}
   We have
    \[\displaystyle{\lim_{\varepsilon \to 0} \, \, \Braket{  {f}^{(\varepsilon)} |  {g}^{(\varepsilon)} } = } -i \frac{\sqrt2 \eta}{1+\eta^2},\] and similar for the other three inner products.
\end{lemma}
\begin{proof}
    Here, too, we make use of the assumption that the preliminary test functions already satisfy the desired inner products, and expand:
 \begin{align*} \left\lvert \Braket{f^{(\varepsilon)} |g^{(\varepsilon)} } + i \frac{\sqrt2 \eta}{1+\eta^2} \right\rvert  &=
\left\lvert \Braket{f^{(\varepsilon)} |g^{(\varepsilon)} } -  \Braket{\widetilde f|\widetilde g} \right\rvert \\
&\leq \frac{1}{\pi} \sum_{n,k,m,\ell}  \Big|f_1(n,k)g_1(m,\ell) - f_2(n,k)g_2(m,\ell) \Big| \notag \\
  \times \,   \bigg\lvert \iint &\left( \frac{1}{x-y} \right) \, \Big(\sigma^\varepsilon_{n,k}(x+\varepsilon) \sigma^\varepsilon_{m,\ell}(y-\varepsilon) - \psi_{n,k}(x)\psi_{m,\ell}(y) \Big) \,\mathrm{d}x\mathrm{d}y \bigg\rvert.
    \end{align*}
    As this is a finite sum, it suffices to show that for each \(n,k ,m,\ell\) contained in the relevant resolution \(\{N_0,N_1,K\}\) (as before, \(k \leq -1\) and \(\ell \geq 0\)), we have \[
    \lim_{\varepsilon\to 0} \,\, \bigg\lvert \iint_{\widetilde{I_{n,k}}\times \widetilde{I_{m,\ell}}} \left( \frac{1}{x-y} \right) \, \Big(\sigma^\varepsilon_{n,k}(x) \sigma^\varepsilon_{m,\ell}(y) - \psi_{n,k}(x)\psi_{m,\ell}(y) \Big) \,\mathrm{d}x\mathrm{d}y \bigg\rvert = 0,
    \]
    where the integrand is (due to the finite translation) supported on the \textit{adapted} domain
\[
 \widetilde{I_{n,k}}\times \widetilde{I_{m,\ell}} = \big[k2^{-n} - \varepsilon, (k+1)2^{-n}\big] \times [\ell 2^{-m}, (l+1) 2^{-m} + \varepsilon ],
\]
  which is  still contained in some \([-M,0] \times [0,M]\) for a sufficiently large \(M > 0\).

    The former limit may now be proved by applying the dominated convergence theorem to the following pointwise limit: \[
 \lim_{\varepsilon \to 0}\,\,   \frac{\sigma^\varepsilon_{n,k}(x) \sigma^\varepsilon_{m,\ell}(y)}{x-y} = \frac{\psi_{n,k}(x) \psi_{m,\ell}(y)}{x-y} \qquad \text{whenever \(x \neq y\)}.
    \]  Since \(\{x\neq y\}\) is the complement of a null set in \(\mathbb R^2\), the dominated convergence theorem can indeed still be applied.
    On the relevant domain, the sequence of functions may be dominated as follows: \[
   \left\lvert  \frac{\sigma^\varepsilon_{n,k}(x+\varepsilon) \sigma^\varepsilon_{m,\ell}(y-\varepsilon)}{x-y} \right\rvert \leq
 \frac{2^{n/2} 2^{m/2} }{\lvert x - y \rvert}, \] where \(|x-y|^{-1}\) is integrable on \(\widetilde{I_{n,k}}\times \widetilde{I_{m,\ell}} \subset [-M,0] \times [0,M]\) just like in the proof of Lemma~\ref{prop:A}. For details, we again refer to Appendix~\ref{app:A}.
\end{proof}

\section{Outlook}
We have provided further evidence of the usefulness of bumpified Haar wavelets, as introduced first in \cite{Dudal2023} and put on firmer formal grounds here, to construct explicit examples of test functions that lead to violations of the Bell-CHSH inequality for free fermion quantum field theories in $d=2$.  We were not able yet to prove
that it can asymptotically furnish the maximal violation (despite strong numerical evidence), cf. Conjecture B. One might notice that the formulation of the problem in terms of \eqref{eq:problem} is general, in the sense that one can look at the matrix representation of the corresponding integral kernel $\frac{1}{x+y}$ in any (orthonormal) basis of choice. Since this kernel resembles that of the Hilbert transform\footnote{Note however that it is not the Hilbert transform, since defined over the half axis and with a subtle sign difference.} which happens to have exactly $\pi$ as $L_2(\mathbb{R})$-norm in our conventions, see e.g.~\cite{stein1959extension,pichorides1972best}, one cannot help but wonder if our Conjecture B could be proven using suitable tools from operator analysis over Hilbert spaces. This is currently under investigation and we hope to come back to this issue in future work.

Although our current approach is less general than the Algebraic QFT Summers-Werner proofs \cite{Summers1987,SummersI1987,SummersII1987}, it does have one big advantage: it is constructive in nature and allows to construct explicit test functions, and by enlarging both resolution and support of the test functions, also stronger and stronger violations. As such, it enables generalizations to the \emph{interacting} QFT cases, which fall completely out of reach of \cite{Summers1987,SummersI1987,SummersII1987}, whose tools are unfortunately limited to free QFTs. Evidently, we do expect that interacting QFTs also exhibit Bell-CHSH violations due to entanglement. A logical first testbed of an interacting QFT would be the $d=2$ Thirring model, which shares its quadratic part with the free case studied here, but supplemented with a 4-fermion interaction. Its Green's functions are known exactly, as one of the very few QFTs, see e.g.~\cite{Johnson:1961cs,Thompson:1983yr,Bozkaya:2005mm}. This implies the exact K\"allén-Lehmann spectral function is also known, and the latter function is of the utmost importance as it enters the definition of the generalized inner products between the test functions. We will report on this in future work, making extensive use of the bumpified Haar wavelets technology.

\section*{Acknowledgments}
 We are grateful to the anonymous Reviewer for urging us to pay careful attention to the construction of the test functions near the closure of the complementary wedges, as well for pointing towards the Hilbert-like transform nature of our problem.

\appendix

\section{\texorpdfstring{$|x+y|^{-1}$}{1/|x+y|} is integrable on \texorpdfstring{$[0,M]\times[0,M]$}{[0,M]×[0,M]}} \label{app:A}

  Fix \(M > 0\). Let us formally prove that the function \[
f\colon \mathbb R^2 \longrightarrow [0,+\infty]\colon (x,y) \longmapsto \frac1{|x+y|}
\] is integrable on  $[0,M]\times[0,M]$ (we define \(f(x,y) =+ \infty\) whenever \(x+y = 0\)).

This fact is used in the proofs of Lemmas~\ref{prop:A} and~\ref{lemma:innerProducts}. Indeed, after a change of variables, it is clear that this is sufficient to formally prove the integrability of $|x+y|^{-1}$ on $[0,M]\times[0,M]$ and of $|x-y|^{-1}$ on $[-M,0]\times[0,M]$.
Since \(f\) is a positive measurable function, it suffices to show the following (well-defined)  integral is finite.

\begin{lemma} We have
$\displaystyle{\iint_{[0,M]\times [0,M]} \frac{1}{x+y} \; \mathrm{d}(x,y)} < \infty$.
\end{lemma}
\begin{proof}
    After a change of variables \(u = x\) and \(v = x+y\), it suffices to prove that
    \[
    \iint_D \frac1v \; \mathrm{d}(u,v) < \infty \quad \text{with} \quad D = \big\{(u,v) \in \mathbb R^2 \mid 0 \leq u \leq M, u \leq v \leq u+ M\big\}.
    \]

    For each \(\varepsilon > 0\), define \(D_\varepsilon = \big\{(u,v) \in \mathbb R^2 \mid \varepsilon \leq u \leq M, u \leq v \leq u+ M\big\}.\) Since \(\frac1v\) is continuous and bounded on \(D_\varepsilon\), it is Riemann integrable, hence certainly
    \[
\iint_{D_\varepsilon} \frac1v \; \mathrm{d}(u,v) < \infty \quad \text{for every \(\varepsilon >0\).}
    \]

    By the monotone convergence theorem, it therefore suffices to prove that
    \[
 \iint_D \frac1v \; \mathrm{d}(u,v) \overset{\text{M.C.T.}}{=} \lim_{\varepsilon \to 0^+}  \iint_{D_\varepsilon}  \frac1v \; \mathrm{d}(u,v) < \infty.
    \]
   Using Fubini's theorem, a routine calculation allows us to deduce that
    \begin{align*}
    \iint_{D_\varepsilon}  &\frac1v \; \mathrm{d}(u,v) = \int_\varepsilon^M \big( \ln(u+M) - \ln u \big) \, \mathrm{d}u = 2 M \ln 2+ M \ln M - (\varepsilon + M)\ln(\varepsilon +M) + \varepsilon \ln \varepsilon.
    \end{align*}
    Therefore indeed,
    \[
\iint_D \frac1v \; \mathrm{d}(u,v) = \lim_{\varepsilon \to 0^+}  \iint_{D_\varepsilon} \! \frac1v \; \mathrm{d}(u,v) = 2M \ln 2 < \infty     ,
    \] completing the proof.
\end{proof}

\bibliography{ms}

\end{document}